\documentclass{article}

\usepackage[latin1]{inputenc}
\usepackage[english]{babel}
\usepackage{jim}
\usepackage{amsmath,amsfonts,amssymb,amsthm,bm,shuffle}
\usepackage[T1]{fontenc}

\usepackage{xcolor}
\definecolor{ColBlack}{RGB}{0,0,0} 
\definecolor{ColWhite}{RGB}{255,255,255} 
\definecolor{ColAA}{HTML}{92C271}
\definecolor{ColAB}{HTML}{509421}
\definecolor{ColAC}{HTML}{406235}
\definecolor{ColBA}{HTML}{EDAE77}
\definecolor{ColBB}{HTML}{C07A3E}
\definecolor{ColBC}{HTML}{8C7460}

\usepackage{tikz}
\usetikzlibrary{shapes}
\usetikzlibrary{fit}
\usetikzlibrary{decorations.pathmorphing}

\usepackage{hyperref}
\usepackage{enumitem}
\usepackage{cite}
\usepackage{abc} 

\usepackage[color=ColAA!40]{todonotes}

\newtheorem{Theorem}{Theorem}[subsection]
\newtheorem{Proposition}[Theorem]{Proposition}

\renewcommand{\leq}{\leqslant}
\renewcommand{\geq}{\geqslant}

\newcommand{\ColAA}[1]{\textcolor{ColAA}{#1}}

\newcommand{\Hide}[1]{\ColAA{\tt HIDEN}}
\newcommand{\Def}[1]{\textcolor{ColBlack!60}{\em #1}}
\newcommand{\Par}[1]{\left(#1\right)}
\newcommand{\Bra}[1]{\left\{#1\right\}}
\newcommand{\Han}[1]{\left[#1\right]}

\tikzstyle{Centering}=[{baseline={([yshift=-0.5ex]current bounding box.center)}}]

\tikzstyle{MarkAA}=[draw=ColAA!80,fill=ColAA!8]
\tikzstyle{MarkAB}=[draw=ColAB!80,fill=ColAB!8]
\tikzstyle{MarkAC}=[draw=ColAC!80,fill=ColAC!8]
\tikzstyle{MarkBA}=[draw=ColBA!80,fill=ColBA!8]
\tikzstyle{MarkBB}=[draw=ColBB!80,fill=ColBB!8]
\tikzstyle{MarkBC}=[draw=ColBC!80,fill=ColBC!8]

\tikzstyle{Node}=[circle,MarkAA,inner sep=1pt,minimum size=2mm,thick,font=\scriptsize]
\tikzstyle{Edge}=[draw=ColBB!80,cap=round,thick,rounded corners=2.5pt]
\tikzstyle{Leaf}=[rectangle,MarkBC,inner sep=0pt,minimum size=1mm,thick]
\tikzstyle{NodeST}=[font=\footnotesize]
\tikzstyle{EdgeLabel}=[midway,inner sep=1pt,fill=ColWhite!0,font=\scriptsize]
\tikzstyle{LeafLabel}=[font=\scriptsize,node distance=2mm]
\tikzstyle{Subtree}=[regular polygon,regular polygon sides=3,MarkA,thick,minimum size=5mm,
font=\scriptsize]

\tikzstyle{PathNode}=[circle,MarkA,thick,inner sep=0pt,minimum size=2mm]
\tikzstyle{PathStep}=[color=Col1!60,thick]

\tikzstyle{Injection}=[ColBlack!100,draw,{>[scale=1.5,length=4,width=5]}-{>[scale=1.5,
length=4,width=5]}]
\tikzstyle{Surjection}=[ColBlack!100,draw,-{>[scale=1.5,length=4,width=5]>[scale=1.5,
length=4,width=5]}]
\tikzstyle{Map}=[ColBlack!100,draw,-{>[scale=1.5,length=4,width=5]}]

\newcommand{\N}{\mathbb{N}}
\newcommand{\Z}{\mathbb{Z}}

\newcommand{\Note}[2]{{#1}_{#2}}
\newcommand{\Notes}{\mathcal{N}}
\newcommand{\Scale}{\bm{\lambda}}
\newcommand{\DegPat}{\mathbf{d}}
\newcommand{\RhyPat}{\mathbf{r}}
\newcommand{\Pat}{\mathbf{p}}
\newcommand{\MPat}{\mathbf{m}}
\newcommand{\CPat}{\mathbf{c}}

\newcommand{\Length}{\ell}
\newcommand{\Mirror}{\mathrm{mir}}

\newcommand{\Rest}{{\square}}
\newcommand{\Beat}{{\blacksquare}}

\newcommand{\Unit}{\mathbf{1}}
\newcommand{\Operad}{\mathcal{O}}
\newcommand{\GeneratingSet}{\mathfrak{G}}
\newcommand{\HadamardProduct}{\boxtimes}

\newcommand{\ConstructionT}{\mathsf{T}}
\newcommand{\OperadDP}{\mathsf{DP}}
\newcommand{\OperadRP}{\mathsf{RP}}
\newcommand{\OperadP}{\mathsf{P}}
\newcommand{\OperadMP}{\OperadP}
\newcommand{\BudOperadMP}{\mathsf{B}\OperadMP}

\newcommand{\SetColors}{\mathfrak{C}}
\newcommand{\Color}{\mathtt{b}}
\newcommand{\ColoredOperad}{\mathcal{C}}
\newcommand{\Out}{\mathrm{out}}
\newcommand{\In}{\mathrm{in}}
\newcommand{\BudOperad}{\mathsf{B}}
\newcommand{\ColoredComposition}{\odot}
\newcommand{\SetRules}{\mathcal{R}}
\newcommand{\BudSystem}{\mathcal{B}}
\newcommand{\PartialDerivation}{\xrightarrow{\circ_i}}
\newcommand{\FullDerivation}{\xrightarrow{\circ}}
\newcommand{\ColoredDerivation}{\xrightarrow{\ColoredComposition}}
\newcommand{\Prune}{\mathrm{pr}}

\newcommand{\BudSystemTemporizator}{\BudSystem^\mathrm{tem}}
\newcommand{\BudSystemRhythmic}{\BudSystem^\mathrm{rhy}}
\newcommand{\BudSystemHarmonizator}{\BudSystem^\mathrm{har}}
\newcommand{\BudSystemArpeggiator}{\BudSystem^\mathrm{arp}}

\newenvironment{MultiPattern}{%
    
    \setlength\arraycolsep{.25em}
    \begin{footnotesize}
        \begin{bmatrix}
        }{
        \end{bmatrix}
    \end{footnotesize}}

\title{Generation of musical patterns through operads}
\date{\today}
\oneauthor{Samuele Giraudo}
{LIGM, univ. Gustave Eiffel, CNRS, ESIEE Paris, ENPC, F-77454 Marne-la-Vall\'ee, France \\
{\tt samuele.giraudo@u-pem.fr}}

\begin{document}
\maketitle

\begin{abstract}
    Nous introduisons la notion de multi-motif, une abstraction combinatoire des phrases
    musicales \`a plusieurs voix.  L'int\'er\^et de cette approche r\'eside dans le fait
    qu'il devient possible de composer deux multi-motifs pour en produire un plus long. Ceci
    s'inscrit dans un contexte alg\'ebrique puisque l'ensemble des multi-motifs poss\`ede
    une structure dite d'op\'erade~; les op\'erades \'etant des structures offrant une
    formalisation de la notion d'op\'erateur et de leurs compositions. Cette vision des
    phrases musicales comme des op\'erateurs permet de r\'ealiser ainsi des calculs sur ces
    derni\`eres et admet des applications en musique g\'en\'erative~: \'etant donn\'e un
    ensemble de courts motifs, nous proposons divers algorithmes pour produire de mani\`ere
    al\'eatoire une nouvelle phrase plus longue inspir\'ee des motifs initiaux.

    \begin{center} --- \end{center}

    \noindent We introduce the notion of multi-pattern, a combinatorial abstraction of
    polyphonic musical phrases.  The interest of this approach lies in the fact that this
    offers a way to compose two multi-patterns in order to produce a longer one. This dives
    musical phrases into an algebraic context since the set of multi-patterns has the
    structure of an operad; operads being structures offering a formalization of the notion
    of operators and their compositions. Seeing musical phrases as operators allows us to
    perform computations on phrases and admits applications in generative music: given a set
    of short patterns, we propose various algorithms to randomly generate a new and longer
    phrase inspired by the inputted patterns.
\end{abstract}

\section*{Introduction}
Generative music is a subfield of computational musicology in which the focus lies on the
automatic creation of musical material.  This creation is based on algorithms accepting
inputs to influence the result obtained, and having a randomized behavior in the sense that
two executions of the algorithm with the same inputs produce different results.  Several
very different approaches exist. For instance, some of them use Markov chains, others
genetic algorithms~\cite{Mat10}, still others neural networks~\cite{BHP20}, or even formal
grammars~\cite{Hol81,HMU06}.  The way in which such algorithms represent and manipulate
musical data is crucial. Indeed, the data structures used to represent musical phrases
orient the nature of the operations we can define of them. Considering operations producing
new phrases from old ones is important to specify algorithms to randomly generate music.  A
possible way for this purpose consists to give at input some musical phrases and the
algorithm creates a new one by blending them through operations. Therefore, the willingness
to endow the infinite set of all musical phrases with operations in order to obtain suitable
algebraic structures is a promising approach. Such interactions between music and algebra is
a fruitful field of investigation~\cite{Mor18,Jed19}.

In this work, we propose to use tools coming from combinatorics and algebraic combinatorics
to represent musical phrases and operations on them, in order to introduce generative
music algorithms close to the family of those based on formal grammars. More precisely, we
introduce the music box model, a very simple model to represent polyphonic phrases, called
multi-patterns. The infinite set of all these objects admits the structure of an operad.
Such structures originate from algebraic topology and are used nowadays also in algebraic
combinatorics and in computer science~\cite{Men15,Gir18}. Roughly speaking, in these
algebraic structures, the elements are operations with several inputs and the composition
law is the  usual composition of operators.  Since the set of multi-patterns forms an
operad, one can regard each pattern as an operation. The fallout of this is that each
pattern is, at the same time, a musical phrase and an operation acting on musical phrases.
In this way, our music box model and its associated operad provide an algebraic and
combinatorial framework to perform computations on musical phrases.

All this admits direct applications to design random generation algorithms since, as
introduced by the author in~\cite{Gir19}, given an operad there exist algorithms to generate
some of its elements. These algorithms are based upon bud generating systems, which are
general formal grammars based on colored operads~\cite{Yau16}. In the present work, we
propose three different variations of these algorithms to produce new musical phrases from
old ones. More precisely, our algorithm works as follows. It takes as input a finite set of
multi-patterns and an integer value to influence the size of the output. It works
iteratively by choosing patterns from the initial collection in order to alter the current
one by performing a composition using the operad structure. As we shall explain, the initial
patterns can be colored in order to forbid some compositions and avoid in this way some
musical intervals for instance. These generation algorithms are not intended to write
complete musical pieces; they are for obtaining, from short old patterns, a similar but
longer one, presenting possibly new ideas to the human composer.

This text is organized as follows. Section~\ref{sec:music_box_model} is devoted to setting
our context and notations about music theory and to introduce the music box model. In
Section~\ref{sec:operads}, we begin by presenting a brief overview of operad theory and we
build step by step the music box operad. For this, we introduce first an operad on sequences
of scale degrees, an operad on rhythm patterns, and then an operad on monophonic patterns to
end with the operad of multi-patterns. Three random generation algorithms for multi-patterns
are introduced in Section~\ref{sec:random_generation}.  Finally,
Section~\ref{sec:applications} provides some concrete applications of the previous
algorithms. We focus here on random variations of a monophonic musical phrase as input
leading to random changes of rhythm, harmonizations, and arpeggiations.

In this version of this work, most of the proofs of the announced results are omitted due to
lack of space.  A computer implementation of all the presented algorithms is, as well as its
source code and {\bf concrete} examples, available at~\cite{Gir20}.

\subsubsection*{General notations and conventions}
For any integer $n$, $[n]$ denotes the set $\{1, \dots, n\}$.  If $a$ is a letter and $n$ is
a nonnegative integer, $a^n$ is the word consisting in $n$ occurrences of $a$. In
particular, $a^0$ is the empty word $\epsilon$.

\section{The music box model} \label{sec:music_box_model}
The purpose of this section is to set some definitions and some conventions about music
theory, and introduce multi-patterns that are abstractions of musical phrases.

\subsection{Notes and scales}
We fit into the context of an $\eta$ tone equal temperament, also written as $\eta$-TET,
where $\eta$ is any nonnegative integer.  An \Def{$\eta$-note} is a pair $(k, n)$ where $0
\leq k \leq \eta - 1$ and $n \in \Z$. We shall write $\Note{k}{n}$ instead of $(k, n)$.  The
integer $n$ is the \Def{octave index} and $k$ is the \Def{step index} of $\Note{k}{n}$.  The
set of all $\eta$-notes is denoted by $\Notes^{(\eta)}$. Despite this level of generality,
and even if all the concepts developed in the sequel work for any~$\eta$, in most
applications and examples we shall consider that $\eta = 12$.  Therefore, under this
convention, we simply call \Def{note} any \Def{$12$-note} and write $\Notes$ for
$\Notes^{(12)}$.  We set in this context of $12$-TET  the ``middle $C$'' as the note
$\Note{0}{4}$, which is the first step of the octave of index~$4$.  An \Def{$\eta$-scale} is
an integer composition $\Scale$ of $\eta$, that is a sequence $\Par{\Scale_1, \dots,
\Scale_\ell}$ of nonnegative integers satisfying
\begin{math}
    \Scale_1 + \dots + \Scale_\ell = \eta.
\end{math}
The \Def{length} of $\Scale$ is the number $\Length(\Scale) := \ell$ of its elements.  We
simply call \Def{scale} any $12$-scale.  For instance, $(2, 2, 1, 2, 2, 2, 1)$ is the major
natural scale, $(2, 1, 2, 2, 1, 3, 1)$ is the harmonic minor scale, and $(2, 1, 4, 1, 4)$ is
the Hirajoshi scale. This encoding of a scale by an integer composition is also known under
the terminology of interval pattern.

A \Def{rooted scale} is a pair $\Par{\Scale, r}$ where $\Scale$ is a scale and $r$ is a
note. This rooted scale describes a subset $\Notes_{\Par{\Scale, r}}$ of $\Notes$ consisting
in the notes reachable from $r$ by following the steps prescribed by the values $\Scale_1$,
$\Scale_2$, \dots, $\Scale_{\Length(\Scale)}$ of $\Scale$.  For instance, if $\Scale$ is the
Hirajoshi scale, then
\begin{equation}
    \Notes_{\Par{\Scale, \Note{0}{4}}}
    = \Bra{\dots,
    \Note{7}{3}, \Note{8}{3},
    {\bf \Note{0}{4}}, 
    \Note{2}{4},
    \Note{3}{4},
    \Note{7}{4},
    \Note{8}{4}, \Note{0}{5}, \dots}.
\end{equation}
If $\Scale$ is the major natural scale, then
\begin{equation}
    \Notes_{\Par{\Scale, \Note{2}{4}}} =
    \Bra{\dots,
     \Note{1}{4},
    {\bf \Note{2}{4}}, \Note{4}{4}, \Note{6}{4}, \Note{7}{4},
    \Note{9}{4}, \Note{11}{4}, \Note{1}{5},
    \Note{2}{5}, \dots}.
\end{equation}

\subsection{Patterns}
We now introduce degree patterns, rhythm patterns, patterns, and finally multi-patterns.

A \Def{degree} $d$ is any element of $\Z$.  Negative degrees are denoted by putting a bar
above their absolute value. For instance, $-3$ is denoted by $\bar{3}$.  A \Def{degree
pattern} $\DegPat$ is a finite word $\DegPat_1 \dots \DegPat_\ell$ of degrees.  The
\Def{arity} of $\DegPat$, also denoted by $|\DegPat|$, is the number $\ell$ of its elements.

Given a rooted scale $(\Scale, r)$, a degree pattern $\DegPat$ specifies a sequence of notes
by assigning to the degree $0$ the note $r$, to the degree $1$ the following higher note in
$\Notes_{\Par{\Scale, r}}$ next to $r$, to the degree $\bar{1}$ the lower note in
$\Notes_{\Par{\Scale, r}}$ next to $r$, and so on.  For instance, the degree pattern $1 0
\bar{2} \bar{3} 5 0 7$ specifies, in the context of the rooted scale $\Par{\Scale,
\Note{0}{4}}$ where $\Scale$ is the major natural scale, the sequence of notes
\begin{equation*}
    \Note{2}{4}, \Note{0}{4}, \Note{9}{3}, \Note{7}{3}, \Note{9}{4},
    \Note{0}{4}, \Note{0}{5}.
\end{equation*}

A \Def{rhythm pattern} $\RhyPat$ is a finite word $\RhyPat_1 \dots \RhyPat_\ell$ on the
alphabet $\Bra{\Rest, \Beat}$.  The symbol $\Rest$ is a \Def{rest} and the symbol $\Beat$ is
a \Def{beat}. The \Def{length} of $\RhyPat$ is $\ell$ and the \Def{arity} $|\RhyPat|$ of
$\RhyPat$ is its number of occurrences of beats.  The \Def{duration sequence} of a rhythm
pattern $\RhyPat$ is the unique sequence $\Par{\alpha_0, \alpha_1, \dots,
\alpha_{|\RhyPat|}}$ of nonnegative integers such that
\begin{equation}
    \RhyPat
    =
    \Rest^{\alpha_0} \enspace \Beat \; \Rest^{\alpha_1}
    \enspace \cdots \enspace
    \Beat \; \Rest^{\alpha_{|\RhyPat|}}.
\end{equation}
The rhythm pattern $\RhyPat$ specifies a rhythm wherein each beat has a relative duration:
the rhythm begins with a silence of $\alpha_0$ units of time, followed by a first beat
sustained $1 + \alpha_1$ units of time, and so on, and finishing by a last beat sustained $1
+ \alpha_{|\RhyPat|}$ units of time. We adopt here the convention that each rest and beat
last each the same amount of time of one eighth of the duration of a whole note. Therefore,
given a tempo specifying how many there are rests and beats by minute, any rhythm pattern
encodes a rhythm.

For instance, let us consider the rhythm pattern
\begin{equation}
    \RhyPat := \Rest \Beat \Beat \Rest \Beat \Rest \Rest \Rest \Beat \Beat \Rest \Beat.
\end{equation}
The duration sequence of $\RhyPat$ is $\Par{1, 0, 1, 3, 0, 1, 0}$ so that $\RhyPat$
specifies the rhythm consisting in an eighth rest, an eighth note, a quarter note, a half
note, an eighth note, a quarter note, and finally an eighth note.

A \Def{pattern} is a pair $\Pat := (\DegPat, \RhyPat)$ such that $|\DegPat| = |\RhyPat|$.
The \Def{arity} $|\Pat|$ of $\Pat$ is the arity of both $\DegPat$ and $\RhyPat$, and the
\Def{length} $\Length(\Pat)$ of $\DegPat$ is the length $\Length(\RhyPat)$ of~$\RhyPat$.

In order to handle concise notations, we shall write any pattern $(\DegPat, \RhyPat)$ as a
word $\Pat$ on the alphabet $\Bra{\Rest} \cup \Z$ where the subword of $\Pat$ obtained by
removing all occurrences of $\Rest$ is the degree pattern $\DegPat$, and the word obtained
by replacing in $\Pat$ each integer by $\Beat$ is the rhythm pattern $\RhyPat$. For
instance,
\begin{equation} \label{equ:example_pattern}
    1 \Rest \Rest \bar{2} \Rest 1 2
\end{equation}
is the concise notation for the pattern
\begin{equation}
    \Par{1 \bar{2} 1 2, \Beat \Rest \Rest \Beat \Rest \Beat \Beat}.
\end{equation}
For this reason, thereafter, we shall see and treat any pattern $\Pat$ as a finite word
$\Pat_1 \dots \Pat_\ell$ on the alphabet $\Bra{\Rest} \cup \Z$. Remark that the length of
$\Pat$ is $\ell$ and that its arity is the number of letters of $\Z$ it has.

Given a rooted scale $(\Scale, r)$ and a tempo, a pattern $\Pat := (\DegPat, \RhyPat)$
specifies a musical phrase, that is a sequence of notes arranged into a rhythm. The notes of
the musical phrase are the ones specified by the degree pattern $\DegPat$ and their relative
durations are specified by the rhythm pattern~$\RhyPat$.
For instance, consider the pattern
\begin{equation}
    \Pat :=
    0 \Rest 1 2 \bar{1} \Rest 0 1 \bar{2} \Rest \bar{1} \bar{0} 0 \Rest \Rest \Rest.
\end{equation}
By choosing the rooted scale $\Par{\Scale, \Note{9}{3}}$ where $\Scale$ is the harmonic
minor scale, and by setting $128$ as tempo, one obtains the musical phrase
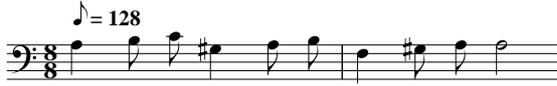
\begin{abc}[name=PhraseExample1,width=.9\abcwidth]
X:1
T:
K:Am
M:8/8
L:1/8
Q:1/8=128
A,2 B, C ^G,2 A, B, | F,2 ^G, A, A,4 |
\end{abc}

For any positive integer $m$, an \Def{$m$-multi-pattern} is an $m$-tuple $\MPat :=
\Par{\MPat^{(1)}, \dots, \MPat^{(m)}}$ of patterns such that all $\MPat^{(i)}$ have the same
arity and the same length. The \Def{arity} $|\MPat|$ of $\MPat$ is the common arity of all
the $\MPat^{(i)}$, and the \Def{length} $\Length(\MPat)$ of $\MPat$ is the common length of
all the $\MPat^{(i)}$.  An $m$-multi-pattern $\MPat$ is denoted through a matrix of
dimension $m \times \Length(\MPat)$, where the $i$-th row contains the pattern $\MPat^{(i)}$
for any $i \in [m]$.  For instance,
\begin{equation}
    \MPat :=
    \begin{MultiPattern}
        0 & \Rest & 1 & \Rest & 1 \\
        \Rest & \bar{2} & \bar{3} & \Rest & 0
    \end{MultiPattern}
\end{equation}
is a $2$-multi-pattern having arity $3$ and length $5$. The fact that all patterns of an
$m$-multi-pattern must have the same length ensures that they last the same amount of units
of time. This is important since an $m$-multi-pattern is used to handle musical sequences
consisting in $m$ stacked voices. The condition about the arities of the patterns, and
hence, about the number of degrees appearing in these, is a particularity of our model and
comes from algebraic reasons. This will be clarified later in the article.

Given a rooted scale $(\Scale, r)$ and a tempo, an $m$-multi-pattern $\MPat$ specifies a
musical phrase obtained by considering the musical phrases specified by each $\MPat^{(i)}$,
$i \in [m]$, each forming a voice.  For instance, consider the $2$-multi-pattern
\begin{equation}
    \MPat :=
    \begin{MultiPattern}
        0 & 4 & \Rest & 4 & 0 & 0 \\
        \bar{7} & \bar{7} & 0 & \Rest & \bar{3} & \bar{3}
    \end{MultiPattern}.
\end{equation}
By choosing the rooted scale $\Par{\Scale, \Note{9}{3}}$ where $\Scale$ is the minor natural
scale and by setting $128$ as tempo, one obtains the musical phrase
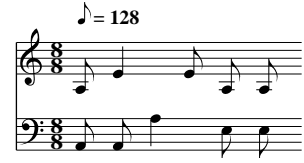
\begin{abc}[name=MultiPatternExample1,width=.47\abcwidth]
X:1
T:
K:Am
M:8/8
L:1/8
Q:1/8=128
V:voice1
A,1 E2 E1 A,1 A,1
V:voice2
A,,1 A,,1 A,2 E,1 E,1
\end{abc}

Due to the fact that $m$-multi-patterns evoke paper tapes of a programmable music box, we
call \Def{music box model} the model just described to represent musical phrases by
$m$-multi-patterns within the context of a rooted scale and a tempo.

\section{Operad structures} \label{sec:operads}
The purpose of this section is to introduce an operad structure on multi-patterns, called
music box operad. The main interest of endowing the set of multi-patterns with the structure
of an operad is that this leads to an algebraic framework to perform computations on
patterns.

\subsection{A primer on operads} \label{subsec:primer_operads}
We set here the elementary notions of operad theory used in the sequel. Most of them come
from~\cite{Gir18}.

A \Def{graded set} is a set $\Operad$ decomposing as a disjoint union
\begin{equation}
    \Operad := \bigsqcup_{n \in \N} \Operad(n),
\end{equation}
where the $\Operad(n)$, $n \in \N$, are sets.  For any $x \in \Operad$, there is by
definition a unique $n \in \N$ such that $x \in \Operad(n)$ called \Def{arity} of $x$ and
denoted by~$|x|$.

A \Def{nonsymmetric operad}, or an \Def{operad} for short, is a triple $\Par{\Operad,
\circ_i, \Unit}$ such that $\Operad$ is a graded set, $\circ_i$ is a map
\begin{equation}
    \circ_i : \Operad(n) \times \Operad(m) \to \Operad(n + m - 1),
    \qquad i \in [n],
\end{equation}
called \Def{partial composition} map, and $\Unit$ is a distinguished element of
$\Operad(1)$, called \Def{unit}. This data has to satisfy, for any $x, y, z \in \Operad$,
the three relations
\begin{equation} \label{equ:operad_axiom_1}
    \Par{x \circ_i y} \circ_{i + j - 1} z = x \circ_i \Par{y \circ_j z},
    \quad i \in [|x|], \enspace j \in [|y|],
\end{equation}
\begin{equation} \label{equ:operad_axiom_2}
    \Par{x \circ_i y} \circ_{j + |y| - 1} z = \Par{x \circ_j z} \circ_i y,
    \quad 1 \leq i < j \leq |x|,
\end{equation}
\begin{equation} \label{equ:operad_axiom_3}
    \Unit \circ_1 x = x = x \circ_i \Unit,
    \quad i \in [|x|].
\end{equation}

Intuitively, an operad is an algebraic structure wherein each element can be seen as an
operator having $|x|$ inputs and one output. Such an operator is depicted as
\begin{equation}
    \begin{tikzpicture}
        [Centering,xscale=.2,yscale=.25,font=\scriptsize]
        \node[NodeST](x)at(0,0){$x$};
        \node(r)at(0,2){};
        \node(x1)at(-3,-2){};
        \node(xn)at(3,-2){};
        \node[below of=x1,node distance=1mm](ex1){$1$};
        \node[below of=xn,node distance=1mm](exn){$|x|$};
        \draw[Edge](r)--(x);
        \draw[Edge](x)--(x1);
        \draw[Edge](x)--(xn);
        \node[below of=x,node distance=6mm]{$\dots$};
    \end{tikzpicture}
\end{equation}
where the inputs are at the bottom and the output at the top.  Given two operations $x$ and
$y$ of $\Operad$, the partial composition $x \circ_i y$ is a new operator obtained by
composing $y$ into $x$ onto its $i$-th input. Pictorially, this partial composition
expresses as
\begin{equation} \label{equ:partial_compostion_on_operators}
    \begin{tikzpicture}[Centering,xscale=.24,yscale=.26,font=\scriptsize]
        \node[NodeST](x)at(0,0){$x$};
        \node(r)at(0,2){};
        \node(x1)at(-3,-2){};
        \node(xn)at(3,-2){};
        \node(xi)at(0,-2){};
        \node[below of=x1,node distance=1mm](ex1){$1$};
        \node[below of=xn,node distance=1mm](exn){$|x|$};
        \node[below of=xi,node distance=1mm](exi){$i$};
        \draw[Edge](r)--(x);
        \draw[Edge](x)--(x1);
        \draw[Edge](x)--(xn);
        \draw[Edge](x)--(xi);
        \node[right of=ex1,node distance=4mm]{$\dots$};
        \node[left of=exn,node distance=4mm]{$\dots$};
    \end{tikzpicture}
    \circ_i
    \begin{tikzpicture}[Centering,xscale=.17,yscale=.25,font=\scriptsize]
        \node[NodeST](x)at(0,0){$y$};
        \node(r)at(0,2){};
        \node(x1)at(-3,-2){};
        \node(xn)at(3,-2){};
        \node[below of=x1,node distance=1mm](ex1){$1$};
        \node[below of=xn,node distance=1mm](exn){$|y|$};
        \draw[Edge](r)--(x);
        \draw[Edge](x)--(x1);
        \draw[Edge](x)--(xn);
        \node[below of=x,node distance=6mm]{$\dots$};
    \end{tikzpicture}
    =
    \begin{tikzpicture}[Centering,xscale=.44,yscale=.3,font=\scriptsize]
        \node[NodeST](x)at(0,0){$x$};
        \node(r)at(0,1.5){};
        \node(x1)at(-3,-2){};
        \node(xn)at(3,-2){};
        \node[below of=x1,node distance=1mm](ex1){$1$};
        \node[below of=xn,node distance=1mm](exn){$|x| + |y| - 1$};
        \node[right of=ex1,node distance=8mm]{$\dots$};
        \node[left of=exn,node distance=9mm]{$\dots$};
        \draw[Edge](r)--(x);
        \draw[Edge](x)--(x1);
        \draw[Edge](x)--(xn);
        \node[NodeST](y)at(0,-2.5){$y$};
        \node(y1)at(-1.6,-4.5){};
        \node(yn)at(1.6,-4.5){};
        \node[below of=y1,node distance=1mm](ey1){$i$};
        \node[below of=yn,node distance=1mm](eyn){\qquad $i + |y| - 1$};
        \draw[Edge](y)--(y1);
        \draw[Edge](y)--(yn);
        \node[below of=y,node distance=7mm]{$\dots$};
        \draw[Edge](x)--(y);
    \end{tikzpicture}.
\end{equation}
Relations~\eqref{equ:operad_axiom_1}, \eqref{equ:operad_axiom_2},
and~\eqref{equ:operad_axiom_3} become clear when they are interpreted into this context of
abstract operators and rooted trees.

Let $\Par{\Operad, \circ_i, \Unit}$ be an operad. The \Def{full composition} map of
$\Operad$ is the map
\begin{equation}
    \circ :
    \Operad(n) \times
    \Operad\Par{m_1}
    \times \dots \times
    \Operad\Par{m_n}
    \to \Operad\Par{m_1 + \dots + m_n},
\end{equation}
defined, for any $x \in \Operad(n)$ and $y_1, \dots, y_n \in \Operad$ by
\begin{equation} \label{equ:full_composition_maps}
    x \circ \Han{y_1, \dots, y_n}
    := \Par{\dots \Par{\Par{x \circ_n y_n} \circ_{n - 1}
        y_{n - 1}} \dots} \circ_1 y_1,
\end{equation}
Intuitively, $x \circ \Han{y_1, \dots, y_n}$ is obtained by grafting simultaneously the
outputs of all the $y_i$ onto the $i$-th inputs of~$x$.

Let $\Par{\Operad', \circ'_i, \Unit'}$ be a second operad. A map $\phi : \Operad \to
\Operad'$ is an \Def{operad morphism} if for any $x \in \Operad(n)$, $\phi(x) \in
\Operad'(n)$, $\phi(\Unit) = \Unit'$, and for any $x, y \in \Operad$ and $i \in [|x|]$,
\begin{equation} \label{equ:operad_morphisms}
    \phi(x \circ_i y) = \phi(x) \circ'_i \phi(y).
\end{equation}
If instead~\eqref{equ:operad_morphisms} holds by replacing the second occurrence of $i$ by
$|x| + 1 - i$, then $\phi$ is an \Def{operad antimorphism}.  We say that $\Operad'$ is a
\Def{suboperad} of $\Operad$ if for any $n \in \N$, $\Operad'(n)$ is a subset of
$\Operad(n)$, $\Unit = \Unit'$, and for any $x, y \in \Operad'$ and $i \in [|x|]$, $x
\circ_i y = x \circ'_i y$. For any subset $\GeneratingSet$ of $\Operad$, the \Def{operad
generated} by $\GeneratingSet$ is the smallest suboperad $\Operad^{\GeneratingSet}$ of
$\Operad$ containing $\GeneratingSet$.  When $\Operad^{\GeneratingSet} = \Operad$ and
$\GeneratingSet$ is minimal with respect to the inclusion among the subsets of
$\GeneratingSet$ satisfying this property, $\GeneratingSet$ is a \Def{minimal generating
set} of $\Operad$ and its elements are \Def{generators} of~$\Operad$.

The \Def{Hadamard product} of $\Operad$ and $\Operad'$ is the operad $\Operad
\HadamardProduct \Operad'$ defined, for any $n \in \N$, by $\Par{\Operad \HadamardProduct
\Operad'}(n) := \Operad(n) \times \Operad'(n)$, endowed with the partial composition map
$\circ''_i$ defined, for any $\Par{x, x'}, \Par{y, y'} \in \Operad \HadamardProduct
\Operad'$ and $i \in \Han{\left|\Par{x, x'}\right|}$, by
\begin{equation}
    \Par{x, x'} \circ''_i \Par{y, y'} :=
    \Par{x \circ_i y, x' \circ'_i y'},
\end{equation}
and having $\Par{\Unit, \Unit'}$ as unit.

\subsection{The music box operad}
We build an operad on multi-patterns step by step by introducing an operad on degree
patterns and an operad on rhythm patterns. The operad of patterns is constructed as the
Hadamard product of the two previous ones. Finally, the operad of multi-patterns if a
suboperad of an iterated Hadamard product of the operad of patterns with itself.

Let $\OperadDP$ be the graded collection of all degree patterns, wherein for any $n \in \N$,
$\OperadDP(n)$ is the set of all degree patterns of arity $n$. Let us define on $\OperadDP$
the partial composition $\circ_i$ wherein, for any degree patterns $\DegPat$ and $\DegPat'$,
and any integer $i \in [|\DegPat|]$,
\begin{footnotesize}
\begin{equation}
    \DegPat \circ_i \DegPat' :=
    \DegPat_1 \dots \DegPat_{i - 1}
    \Par{\DegPat_i + \DegPat'_1} \dots \Par{\DegPat_i + \DegPat'_{|\DegPat'|}}
    \DegPat_{i + 1} \dots \DegPat_{|\DegPat|}.
\end{equation}
\end{footnotesize}
For instance,
\begin{equation}
    0 {\bf 1} 234 \circ_2 {\bf \bar{1}10} = 0 {\bf 021} 234.
\end{equation}
We denote by $\epsilon$ the empty degree pattern. This element is the only one
of~$\OperadDP(0)$.

\begin{Proposition} \label{prop:operad_degree_patterns}
    The triple $\Par{\OperadDP, \circ_i, 0}$ is an operad.
\end{Proposition}
\begin{proof}
    This is the consequence of the fact that $\Par{\OperadDP, \circ_i, 0}$ is the image of
    the monoid $\Par{\Z, +, 0}$ by the construction $\ConstructionT$ defined
    in~\cite{Gir15}. Since this construction associates an operad with any monoid, the
    result follows.
\end{proof}

We call $\OperadDP$ the \Def{degree pattern operad}.

\begin{Proposition} \label{prop:generating_set_operad_degree_patterns}
    The operad $\OperadDP$ admits $\Bra{\epsilon, \bar{1}, 1, 00}$ and $\Bra{\epsilon,
    \bar{1} 1}$ as minimal generating sets.
\end{Proposition}

Let $\OperadRP$ be the graded collection of all rhythm patterns, wherein for any $n \in \N$,
$\OperadRP(n)$ is the set of all rhythm patterns of arity $n$. Let us define on $\OperadRP$
the partial composition $\circ_i$ wherein, for any rhythm patterns $\RhyPat$ and $\RhyPat'$,
and any integer $i \in [|\RhyPat|]$, $\RhyPat \circ_i \RhyPat'$ is obtained by replacing the
$i$-th occurrence of $\Beat$ in $\RhyPat$ by $\RhyPat'$.  For instance,
\begin{equation}
    \Beat \Beat \Rest \Beat \Rest \Rest \Beat
    \circ_3
    \Rest \Beat \Rest \Beat
    =
    \Beat \Beat \Rest \; \Rest \Beat \Rest \Beat \; \Rest \Rest \Beat.
\end{equation}
We denote by $\epsilon$ the empty rhythm pattern. This element is not the only one of
$\OperadRP(0)$ since $\OperadRP(0) = \Bra{\Rest^\alpha : \alpha \in \N}$.

\begin{Proposition} \label{prop:operad_rhythm_patterns}
    The triple $\Par{\OperadRP, \circ_i, \Beat}$ is an operad.
\end{Proposition}

We call $\OperadRP$ the \Def{rhythm pattern operad}.

\begin{Proposition} \label{prop:generating_set_operad_rhythm_pattern}
    The operad $\OperadRP$ admits $\Bra{\epsilon, \Rest, \Beat \Beat}$ as minimal generating
    set.
\end{Proposition}

Let $\OperadP$ be the operad defined as
\begin{equation}
    \OperadP := \OperadDP \HadamardProduct \OperadRP.
\end{equation}
Since a pattern is a pair $(\DegPat, \RhyPat)$ where $\DegPat$ is a degree pattern and
$\RhyPat$ is a rhythm pattern of the same arity, for any $n \in \N$, $\OperadP(n)$ is in
fact the set of all patterns of arity $n$. For this reason, $\OperadP$ is the graded set of
all patterns. We call $\OperadP$ the \Def{pattern operad}.  For instance, by using the
concise notation for patterns,
\begin{equation}
    \Rest \bar{2} {\bf 1} \Rest 1 \circ_2 {\bf 0 \Rest \Bar{1}}
    = \Rest \bar{2} \; {\bf 1 \Rest 0} \; \Rest 1.
\end{equation}
We denote by $\epsilon$ the empty pattern.

\begin{Proposition} \label{prop:generating_set_operad_patterns}
    The operad $\OperadP$ admits $\Bra{\epsilon, \Rest, \bar{1}, 1, 00}$ and $\Bra{\epsilon,
    \Rest, \bar{1} 1}$ as minimal generating sets.
\end{Proposition}

A consequence of Proposition~\ref{prop:generating_set_operad_patterns} is that any pattern
$\Pat$ expresses as a tree on the internal nodes in $\Bra{\epsilon, \Rest, \bar{1}, 1, 00}$
or in $\Bra{\epsilon, \Rest, \bar{1} 1}$. For instance, the pattern
$\Pat := \bar{1} \Rest \Rest 1 \Rest 3$ expresses as the trees
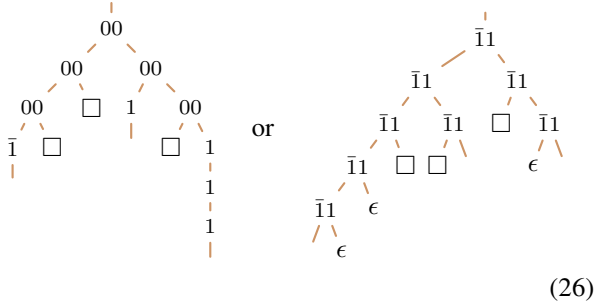
\begin{equation}
    \begin{tikzpicture}[Centering,xscale=0.26,yscale=0.23]
        \node(1)at(0.00,-9.14){};
        \node(10)at(8.00,-6.86){$\Rest$};
        \node(15)at(10.00,-13.71){};
        \node(3)at(2.00,-6.86){$\Rest$};
        \node(5)at(4.00,-4.57){$\Rest$};
        \node(8)at(6.00,-6.86){};
        \node[NodeST](0)at(0.00,-6.86){$\bar{1}$};
        \node[NodeST](11)at(9.00,-4.57){$00$};
        \node[NodeST](12)at(10.00,-6.86){$1$};
        \node[NodeST](13)at(10.00,-9.14){$1$};
        \node[NodeST](14)at(10.00,-11.43){$1$};
        \node[NodeST](2)at(1.00,-4.57){$00$};
        \node[NodeST](4)at(3.00,-2.29){$00$};
        \node[NodeST](6)at(5.00,0.00){$00$};
        \node[NodeST](7)at(6.00,-4.57){$1$};
        \node[NodeST](9)at(7.00,-2.29){$00$};
        \draw[Edge](0)--(2);
        \draw[Edge](1)--(0);
        \draw[Edge](10)--(11);
        \draw[Edge](11)--(9);
        \draw[Edge](12)--(11);
        \draw[Edge](13)--(12);
        \draw[Edge](14)--(13);
        \draw[Edge](15)--(14);
        \draw[Edge](2)--(4);
        \draw[Edge](3)--(2);
        \draw[Edge](4)--(6);
        \draw[Edge](5)--(4);
        \draw[Edge](7)--(9);
        \draw[Edge](8)--(7);
        \draw[Edge](9)--(6);
        \node(r)at(5.00,2){};
        \draw[Edge](r)--(6);
    \end{tikzpicture}
    \quad
    \mbox{or}
    \quad
    \begin{tikzpicture}[Centering,xscale=0.21,yscale=0.20]
        \node(0)at(0.00,-14.17){};
        \node(10)at(10.00,-8.50){};
        \node(12)at(12.00,-5.67){$\Rest$};
        \node(14)at(14.00,-8.50){$\epsilon$};
        \node(16)at(16.00,-8.50){};
        \node(2)at(2.00,-14.17){$\epsilon$};
        \node(4)at(4.00,-11.33){$\epsilon$};
        \node(6)at(6.00,-8.50){$\Rest$};
        \node(8)at(8.00,-8.50){$\Rest$};
        \node[NodeST](1)at(1.00,-11.33){$\bar{1} 1$};
        \node[NodeST](11)at(11.00,0.00){$\bar{1} 1$};
        \node[NodeST](13)at(13.00,-2.83){$\bar{1} 1$};
        \node[NodeST](15)at(15.00,-5.67){$\bar{1} 1$};
        \node[NodeST](3)at(3.00,-8.50){$\bar{1} 1$};
        \node[NodeST](5)at(5.00,-5.67){$\bar{1} 1$};
        \node[NodeST](7)at(7.00,-2.83){$\bar{1} 1$};
        \node[NodeST](9)at(9.00,-5.67){$\bar{1} 1$};
        \draw[Edge](0)--(1);
        \draw[Edge](1)--(3);
        \draw[Edge](10)--(9);
        \draw[Edge](12)--(13);
        \draw[Edge](13)--(11);
        \draw[Edge](14)--(15);
        \draw[Edge](15)--(13);
        \draw[Edge](16)--(15);
        \draw[Edge](2)--(1);
        \draw[Edge](3)--(5);
        \draw[Edge](4)--(3);
        \draw[Edge](5)--(7);
        \draw[Edge](6)--(5);
        \draw[Edge](7)--(11);
        \draw[Edge](8)--(9);
        \draw[Edge](9)--(7);
        \node(r)at(11.00,2.25){};
        \draw[Edge](r)--(11);
    \end{tikzpicture}
\end{equation}
respectively for the two previous generating sets.

For any positive integer $m$, let $\OperadMP'_m$ be operad defined through the iterated
Hadamard product
\begin{equation}
    \OperadMP'_m :=
    \underbrace{\OperadP \HadamardProduct
    \cdots \HadamardProduct \OperadP}_{\footnotesize m \mbox{ terms}}.
\end{equation}
Let also $\OperadMP_m$ be the subset of $\OperadMP'_m$ restrained on the $m$-tuples
$\Par{\MPat^{(1)}, \dots, \MPat^{(m)}}$ such that $\Length\Par{\MPat^{(1)}} = \dots =
\Length\Par{\MPat^{(m)}}$.

\begin{Theorem} \label{thm:operad_multi_patterns}
    For any positive integer $m$, $\OperadMP_m$ is an operad.
\end{Theorem}

Since an $m$-multi-pattern is an $m$-tuple $\Par{\MPat^{(1)}, \dots, \MPat^{(m)}}$ where all
$\MPat^{(i)}$ have the same arity and the same length, for any $m \in \N$, $\OperadMP_m$ is
the graded set of all $m$-multi-patterns. By Theorem~\ref{thm:operad_multi_patterns},
$\OperadMP_m$ is an operad, called \Def{$m$-music box operad}.

By using the matrix notation for $m$-multi-patterns, we have for instance respectively in
$\OperadMP_2$ and in $\OperadMP_3$,
\begin{equation}
    \begin{MultiPattern}
        \Rest & \bar{2} & \bar{1} & \Rest & 0 \\
        0 & 1 & \Rest & \Rest & 1
    \end{MultiPattern}
    \circ_2
    \begin{MultiPattern}
        {\bf 1} & \Rest & {\bf 0} & {\bf 0} \\
        {\bf \bar{3}} & \Rest & {\bf 0} & {\bf 4} \\
    \end{MultiPattern}
    =
    \begin{MultiPattern}
        \Rest & \bar{2} & {\bf 0} & \Rest & {\bf \bar{0}} & {\bf \bar{0}} & \Rest & 0 \\
        0 & {\bf \bar{2}} & \Rest & {\bf 1} & {\bf 5} & \Rest & \Rest & 1
    \end{MultiPattern}.
\end{equation}
This definition of the $m$-music box operad $\OperadMP_m$ explains why all the patterns of
an $m$-multi-pattern must have the same arity. This is a consequence of the general
definition of the Hadamard product of operads.

For any sequence $\Par{\alpha_1, \dots, \alpha_m}$ of integers of $\Z$ and $\beta \in \N$,
let
\begin{equation}
    \phi_{\Par{\alpha_1, \dots, \alpha_m}, \beta} : \OperadMP_m \to \OperadMP_m
\end{equation}
be the map such that, for any $\MPat := \Par{\MPat^{(1)}, \dots, \MPat^{(m)}} \in
\OperadMP_m$,
\begin{math}
    \phi_{\Par{\alpha_1, \dots, \alpha_m}, \beta}\Par{\MPat}
\end{math}
is the $m$-multi-pattern obtained by multiplying each degree of $\MPat^{(j)}$ by $\alpha_j$
and by replacing each occurrence of $\Rest$ in $\MPat$ by $\beta$ occurrences of $\Rest$.
For instance,
\begin{equation}
    \phi_{\Par{2, 0, -1}, 2}\Par{
    \begin{MultiPattern}
        1 & \Rest & \Rest & 2 \\
        \Rest & 1 & \Rest & 3 \\
        3 & 1 & \Rest & \Rest
    \end{MultiPattern}}
    =
    \begin{MultiPattern}
        2 & \Rest & \Rest & \Rest & \Rest & 4 \\
        \Rest & \Rest & 0 & \Rest & \Rest & 0 \\
        \bar{3} & \bar{1} & \Rest & \Rest & \Rest & \Rest
    \end{MultiPattern}.
\end{equation}

\begin{Proposition} \label{prop:morphism_operad_multi_patterns}
    For any positive integer $m$, any sequence $\Par{\alpha_1, \dots, \alpha_m}$ of
    integers, and any nonnegative integer $\beta$, the map $\phi_{\Par{\alpha_1, \dots,
    \alpha_m}, \beta}$ is an operad endomorphism of~$\OperadMP_m$.
\end{Proposition}

Let also the map $\Mirror : \OperadMP_m \to \OperadMP_m$ be the map such that, for any
$\MPat \in \OperadMP_m$, $\Mirror(\MPat)$ is the $m$-multi-pattern obtained by reading the
$\MPat$ from right to left. For instance,
\begin{equation}
    \Mirror\Par{
    \begin{MultiPattern}
        1 & \Rest & \Rest & 2 \\
        \Rest & 1 & \Rest & 3 \\
        3 & 1 & \Rest & \Rest
    \end{MultiPattern}}
    =
    \begin{MultiPattern}
        2 & \Rest & \Rest & 1 \\
        3 & \Rest & 1 & \Rest \\
        \Rest & \Rest & 1 & 3
    \end{MultiPattern}.
\end{equation}

\begin{Proposition} \label{prop:mirror_operad_multi_patterns}
    For any positive integer $m$, the map $\Mirror$ sending any $m$-multi-pattern to its
    mirror is an operad anti-automorphism of $\OperadMP_m$.
\end{Proposition}

Due to the $m$-music box operad and more precisely, to the operad structure on
$m$-multi-patterns, we can see any $m$-multi-pattern as an operator. Therefore, we can build
$m$-multi-patterns and then musical sequences by considering some compositions of small
building blocks $m$-multi-patterns. For instance, by considering the small
$2$-multi-patterns
\begin{equation}
    \MPat_1 :=
    \begin{MultiPattern}
        0 & \Rest \\
        \Rest & 0
    \end{MultiPattern},
    \enspace
    \MPat_2 :=
    \begin{MultiPattern}
        1 & 0 & 1 \\
        \bar{7} & 0 & 0
    \end{MultiPattern},
    \enspace
    \MPat_3 :=
    \begin{MultiPattern}
        1 & 2 & \Rest & 3 \\
        \bar{1} & 0 & \Rest & 1
    \end{MultiPattern},
\end{equation}
one can build a new $2$-multi-pattern by composing them as specified by the tree
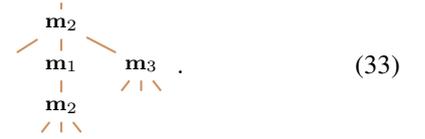
\begin{equation} \label{equ:example_composition_tree_multi_patterns}
    \begin{tikzpicture}[Centering,xscale=0.35,yscale=0.2]
        \node(0)at(0.00,-2.25){};
        \node(10)at(6.00,-5){};
        \node(3)at(1.00,-7.75){};
        \node(5)at(2.00,-7.75){};
        \node(6)at(3.00,-7.75){};
        \node(7)at(4.00,-5){};
        \node(9)at(5.00,-5){};
        \node[NodeST](1)at(2.00,0.00){$\MPat_2$};
        \node[NodeST](2)at(2.00,-2.75){$\MPat_1$};
        \node[NodeST](4)at(2.00,-5.50){$\MPat_2$};
        \node[NodeST](8)at(5.00,-2.75){$\MPat_3$};
        \draw[Edge](0)--(1);
        \draw[Edge](10)--(8);
        \draw[Edge](2)--(1);
        \draw[Edge](3)--(4);
        \draw[Edge](4)--(2);
        \draw[Edge](5)--(4);
        \draw[Edge](6)--(4);
        \draw[Edge](7)--(8);
        \draw[Edge](8)--(1);
        \draw[Edge](9)--(8);
        \node(r)at(2.00,2){};
        \draw[Edge](r)--(1);
    \end{tikzpicture}.
\end{equation}
This produces the new $2$-multi-pattern
\begin{equation} \label{equ:example_multi_pattern}
    \begin{MultiPattern}
       1 & 1 & 0 & 1 & \Rest & 2 & 3 & \Rest & 3 \\
       \bar{7} & \Rest & \bar{7} & 0 & 0 & \bar{1} & 0 & \Rest & 1
    \end{MultiPattern}.
\end{equation}
Besides, by Proposition~\ref{prop:morphism_operad_multi_patterns}, the image
of~\eqref{equ:example_multi_pattern} through the map, for instance,  $\phi_{\Par{-1, 2}, 3}$
is the same as the $2$-multi-pattern obtained
from~\eqref{equ:example_composition_tree_multi_patterns} by replacing each
$2$-multi-patt\-ern appearing in it by its image by~$\phi_{\Par{-1, 2}, 3}$.

\section{Generation and random generation} \label{sec:random_generation}
We exploit now the music box operad to design three random generation algorithms devoted to
generate new musical phrases from a finite set of multi-patterns. This relies on colored
operads and bud generating systems, a sort of formal grammars introduced in~\cite{Gir19}.

\subsection{Colored operads and bud operads}
We provide here the elementary notions about colored operads~\cite{Yau16}. We also explain
how to build colored operads from an operads.

A \Def{set of colors} is any nonempty finite set
\begin{math}
    \SetColors := \Bra{\Color_1, \dots, \Color_k}
\end{math}
wherein elements are called \Def{colors}.  A \Def{$\SetColors$-colored set} is a set
$\ColoredOperad$ decomposing as a disjoint union
\begin{equation}
    \ColoredOperad
    :=
    \bigsqcup_{\substack{
        a \in \SetColors \\
        u \in \SetColors^*}}
    \ColoredOperad(a, u),
\end{equation}
where $\SetColors^*$ is the set of all finite sequences of elements of $\SetColors$, and the
$\ColoredOperad(a, u)$ are sets. For any $x \in \ColoredOperad$, there is by definition a
unique pair $(a, u) \in \SetColors \times \SetColors^*$ such that $x \in \ColoredOperad(a,
u)$.  The \Def{arity} $|x|$ of $x$ is the length $|u|$ of $u$ as a word, the \Def{output
color} $\Out(x)$ of $x$ is $a$, and for any $i \in [|x|]$, the \Def{$i$-th input color}
$\In_i(x)$ of $x$ is the $i$-th letter $u_i$ of $u$. We also denote, for any $n \in \N$, by
$\ColoredOperad(n)$ the set of all elements of $\ColoredOperad$ of arity $n$. Therefore, a
colored graded set is in particular a graded set.

A \Def{$\SetColors$-colored operad} is a triple $\Par{\ColoredOperad, \circ_i, \Unit}$
such that $\ColoredOperad$ is a $\SetColors$-colored set, $\circ_i$ is a map
\begin{equation}
    \circ_i : \ColoredOperad(a, u) \times \ColoredOperad\Par{u_i, v}
    \to \ColoredOperad\Par{a, u \circ_i v},
    \qquad i \in [|u|],
\end{equation}
called \Def{partial composition} map, where $u \circ_i v$ is the word on $\SetColors$
obtained by replacing the $i$-th letter of $u$ by $v$, and $\Unit$ is a map
\begin{equation}
    \Unit : \SetColors \to \ColoredOperad(a, a),
\end{equation}
called \Def{colored unit} map.  This data has to satisfy
Relations~\eqref{equ:operad_axiom_1} and~\eqref{equ:operad_axiom_2} when their left and
right members are both well-defined, and, for any $x \in \ColoredOperad$, the relation
\begin{equation}
    \Unit(\Out(x)) \circ_1 x = x = x \circ_i \Unit\Par{\In_i(x)},
    \qquad i \in[|x|].
\end{equation}

Intuitively, an element $x$ of a colored operad having $a$ as output color and $u_i$ as
$i$-th input color for any $i \in [|x|]$ can be seen as an abstract operator wherein colors
are assigned to its output and to each of its inputs. Such an operator is depicted as
\begin{equation}
    \begin{tikzpicture}[Centering,xscale=.25,yscale=.33,font=\scriptsize]
        \node[NodeST](x)at(0,0){$x$};
        \node(r)at(0,2){};
        \node(x1)at(-3,-2){};
        \node(xn)at(3,-2){};
        \node[below of=x1,node distance=1mm](ex1){$1$};
        \node[below of=xn,node distance=1mm](exn){$|x|$};
        \draw[Edge](r)edge[]node[EdgeLabel]{$a$}(x);
        \draw[Edge](x)edge[]node[EdgeLabel]{$u_1$}(x1);
        \draw[Edge](x)edge[]node[EdgeLabel]{$u_{|x|}$}(xn);
        \node[below of=x,node distance=8mm]{$\dots$};
    \end{tikzpicture},
\end{equation}
where the colors of the output and inputs are put on the corresponding edges. The partial
composition of two elements $x$ and $y$ in a colored operad expresses pictorially as
\begin{equation}
    \begin{tikzpicture}[Centering,xscale=.3,yscale=.35,font=\scriptsize]
        \node[NodeST](x)at(0,0){$x$};
        \node(r)at(0,1.75){};
        \node(x1)at(-3,-2){};
        \node(xn)at(3,-2){};
        \node(xi)at(0,-2){};
        \node[below of=x1,node distance=1mm](ex1){$1$};
        \node[below of=xn,node distance=1mm](exn){$|x|$};
        \node[below of=xi,node distance=1mm](exi){$i$};
        \draw[Edge](r)edge[]node[EdgeLabel]{$a$}(x);
        \draw[Edge](x)edge[]node[EdgeLabel]{$u_1$}(x1);
        \draw[Edge](x)edge[]node[EdgeLabel]{$u_{|x|}$}(xn);
        \draw[Edge](x)edge[]node[EdgeLabel]{$u_i$}(xi);
        \node[right of=ex1,node distance=5mm]{$\dots$};
        \node[left of=exn,node distance=5mm]{$\dots$};
    \end{tikzpicture}
    \circ_i
    \begin{tikzpicture}[Centering,xscale=.2,yscale=.35,font=\scriptsize]
        \node[NodeST](x)at(0,0){$y$};
        \node(r)at(0,2){};
        \node(x1)at(-3,-2){};
        \node(xn)at(3,-2){};
        \node[below of=x1,node distance=1mm](ex1){$1$};
        \node[below of=xn,node distance=1mm](exn){$|y|$};
        \draw[Edge](r)edge[]node[EdgeLabel]{$u_i$}(x);
        \draw[Edge](x)edge[]node[EdgeLabel]{$v_1$}(x1);
        \draw[Edge](x)edge[]node[EdgeLabel]{$v_{|y|}$}(xn);
        \node[below of=x,node distance=8mm]{$\dots$};
    \end{tikzpicture}
    =
    \begin{tikzpicture}[Centering,xscale=.41,yscale=.45,font=\scriptsize]
        \node[NodeST](x)at(0,0){$x$};
        \node(r)at(0,1.5){};
        \node(x1)at(-3,-2){};
        \node(xn)at(3,-2){};
        \node[below of=x1,node distance=1mm](ex1){$1$};
        \node[below of=xn,node distance=1mm](exn){$|x| \! + \! |y| \! - \! 1$};
        \node[right of=ex1,node distance=8mm]{$\dots$};
        \node[left of=exn,node distance=9mm]{$\dots$};
        \draw[Edge](r)edge[]node[EdgeLabel]{$a$}(x);
        \draw[Edge](x)edge[]node[EdgeLabel]{$u_1$}(x1);
        \draw[Edge](x)edge[]node[EdgeLabel]{$u_{|x|}$}(xn);
        \node[NodeST](y)at(0,-2.5){$y$};
        \node(y1)at(-1.6,-4.5){};
        \node(yn)at(1.6,-4.5){};
        \node[below of=y1,node distance=1mm](ey1){$i$};
        \node[below of=yn,node distance=1mm](eyn){\quad $i \! + \! |y| \! - \! 1$};
        \draw[Edge](y)edge[]node[EdgeLabel]{$v_1$}(y1);
        \draw[Edge](y)edge[]node[EdgeLabel]{$v_{|y|}$}(yn);
        \node[below of=y,node distance=10mm]{$\dots$};
        \draw[Edge](x)edge[]node[EdgeLabel]{$u_i$}(y);
    \end{tikzpicture}.
\end{equation}
Besides, most of the definitions about operads recalled in
Section~\ref{subsec:primer_operads} generalize straightforwardly to
colored operads. In particular, one can consider the full composition map of a colored
operad defined by~\eqref{equ:full_composition_maps} when its right member is well-defined.

Let us introduce another operation, specific to colored operads. Let $\Par{\ColoredOperad,
\circ_i, \Unit}$ be a colored operad. The \Def{colored composition} map of
$\ColoredOperad$ is the map
\begin{equation}
    \ColoredComposition :
    \ColoredOperad(a, u) \times \ColoredOperad(b, v) \to \ColoredOperad,
    \quad
    a, b \in \SetColors,
    \enspace
    u, v \in \SetColors^*,
\end{equation}
defined, for any $x \in \ColoredOperad(a, u)$ and $y \in \ColoredOperad(b, v)$, by using the
full composition map, by
\begin{equation}
    x \ColoredComposition y := x \circ \Han{y^{(1)}, \dots, y^{(|x|)}},
\end{equation}
where for any $i \in [|x|]$,
\begin{equation}
    y^{(i)} :=
    \begin{cases}
        y & \mbox{if } \In_i(x) = \Out(y), \\
        \Unit\Par{\In_i(x)} & \mbox{otherwise}.
    \end{cases}
\end{equation}
Intuitively, $x \ColoredComposition y$ is obtained by grafting simultaneously the outputs of
copies of $y$ into all the inputs of $x$ having the same color as the output color of~$y$.

Let us describe a general construction building a colored operad from a noncolored one
introduced in~\cite{Gir19}.  Given a noncolored operad $\Par{\Operad, \circ_i, \Unit}$ and a
set of colors $\SetColors$, the \Def{$\SetColors$-bud operad} of $\Operad$ is the
$\SetColors$-colored operad $\BudOperad_\SetColors(\Operad)$ defined in the following way.
First, $\BudOperad_\SetColors(\Operad)$ is the $\SetColors$-colored set defined, for any $a
\in \SetColors$ and $u \in \SetColors^*$, by
\begin{equation}
    \BudOperad_\SetColors(\Operad)(a, u) := \Bra{(a, x, u) : x \in \Operad(|u|)}.
\end{equation}
Second, the partial composition maps $\circ_i$ of $\BudOperad_\SetColors(\Operad)$ are
defined, for any $(a, x, u), \Par{u_i, y, v} \in \BudOperad_\SetColors(\Operad)$ and $i \in
[|u|]$, by
\begin{equation} \label{equ:partial_composition_map_bud_operad}
    (a, x, u) \circ_i \Par{u_i, y, v} := \Par{a, x \circ_i y, u \circ_i v}
\end{equation}
where the first occurrence of $\circ_i$ in the right member
of~\eqref{equ:partial_composition_map_bud_operad} is the partial composition map of $\Operad$ and the
second one is a substitution of words: $u \circ_i v$ is the word obtained by replacing in
$u$ the $i$-th letter of $u$ by $v$.  Finally, the colored unit map $\Unit$ of
$\BudOperad_\SetColors(\Operad)$ is defined by $\Unit(a) := \Par{a, \Unit, a}$ for any $a
\in \SetColors$, where $\Unit$ is the unit of $\Operad$. The \Def{pruning} $\Prune((a, x, u))$ of an element $(a, x, u)$ of
$\BudOperad_\SetColors(\Operad)$ is the element $x$ of~$\Operad$.

Intuitively, this construction consists in forming a colored operad
$\BudOperad_\SetColors(\Operad)$ out of $\Operad$ by surrounding its elements with an output
color and input colors coming from $\SetColors$ in all possible ways.

We apply this construction to the $m$-music box operad by setting, for any set $\SetColors$
of colors,
\begin{equation}
    \BudOperadMP_m^\SetColors := \BudOperad_\SetColors\Par{\OperadMP_m}.
\end{equation}
We call $\BudOperadMP_m^\SetColors$ the \Def{$\SetColors$-bud $m$-music box operad}. The
elements of $\BudOperadMP_m^\SetColors$ are called \Def{$\SetColors$-colored
$m$-multi-patterns}. For instance, for $\SetColors := \Bra{\Color_1, \Color_2, \Color_3}$,
\begin{equation}
    \Par{\Color_1,
    \begin{MultiPattern}
        1 & \Rest & 0 & \Rest & 1 \\
        \bar{7} & \Rest & 0 & 0 & \Rest
    \end{MultiPattern},
    \Color_2 \Color_2 \Color_1}
\end{equation}
is a $\SetColors$-colored $2$-multi-pattern. Moreover, in the colored operad
$\BudOperadMP_2^\SetColors$, one has
\begin{multline}
    \Par{\Color_3,
    \begin{MultiPattern}
        0 & 1 & \Rest \\
        \bar{1} & \Rest & 0
    \end{MultiPattern},
    \Color_2 \Color_1}
    \circ_2
    \Par{\Color_1,
    \begin{MultiPattern}
        {\bf 1} & {\bf 1} & {\bf 2} \\
        {\bf 2} & {\bf \bar{1}} & {\bf \bar{2}}
    \end{MultiPattern},
    {\bf \Color_3 \Color_3 \Color_2}}
    \\
    =
    \Par{\Color_3,
    \begin{MultiPattern}
        0 & {\bf 2} & {\bf 2} & {\bf 3} & \Rest \\
        \bar{1} & \Rest & {\bf 2} & {\bf \bar{1}} & {\bf \bar{2}}
    \end{MultiPattern},
    \Color_2 {\bf \Color_3 \Color_3 \Color_2}}.
\end{multline}

The intuition that justifies the introduction of these colored versions of patterns and of
the $m$-music box operad is that colors restrict the right to perform the composition of two
given patterns. In this way, one can for instance forbid some intervals in the musical
phrases specified by the patterns of a suboperad of $\BudOperadMP_m^\SetColors$ generated by
a given set of $\SetColors$-colored $m$-multi-patterns.  Moreover, given a set
$\GeneratingSet$ of $\SetColors$-colored $m$-multi-patterns, the elements of the suboperad
${\BudOperadMP_m^\SetColors}^\GeneratingSet$ of $\BudOperadMP_m^\SetColors$ generated by
$\GeneratingSet$ are obtained by composing elements of $\GeneratingSet$. Therefore, in some
sense, these elements inherit from properties of the patterns~$\GeneratingSet$.

The next section uses these ideas to propose random generation algorithms outputting new
patterns from existing ones in a controlled way.

\subsection{Bud generating systems and random generation}
\label{subsec:bud_generating_systems}
We describe here a sort of generating systems using operads introduced in~\cite{Gir19}.
Slight variations are considered in this present work. We also design three random
generation algorithms to produce musical phrases.

A \Def{bud generating system}~\cite{Gir19} is a tuple $\Par{\Operad, \SetColors, \SetRules,
\Color}$ where
\begin{enumerate}[label={\em (\roman*)}]
    \item $\Par{\Operad, \circ_i, \Unit}$ is a noncolored operad, called \Def{ground operad};
    \item $\SetColors$ is a finite set of colors;
    \item $\SetRules$ is a finite subset of $\BudOperad_\SetColors(\Operad)$, called
    \Def{set of rules};
    \item $\Color$ is a color of $\SetColors$, called \Def{initial color}.
\end{enumerate}
For any color $a \in \SetColors$, we shall denote by $\SetRules_a$ the set of all rules of
$\SetRules$ having $a$ as output color.

Bud generating systems are devices similar to context-free formal grammars~\cite{HMU06}
wherein colors play the role of nonterminal symbols. These last devices are designed to
generate sets of words. Bud generating systems are designed to generate more general
combinatorial objects (here, $m$-multi-patterns). More precisely, a bud generating system
$\Par{\Operad, \SetColors, \SetRules, \Color}$ allows us to build elements of $\Operad$ by
following three different operating modes. We describe in the next sections the three
corresponding random generation algorithms. These algorithms are in particular intended to
work with $\OperadMP_m$ as ground operad in order to generate $m$-multi-patterns.

Hereafter, we shall provide some examples based upon the bud generating system
\begin{equation} \label{equ:example_bud_generating_system}
    \BudSystem :=
    \Par{\OperadMP_2, \Bra{\Color_1, \Color_2, \Color_3},
    \Bra{\CPat_1, \CPat_2, \CPat_3, \CPat_4, \CPat_5}, \Color_1}
\end{equation}
where
\begin{equation}
    \CPat_1 :=
    \Par{\Color_1,
    \begin{MultiPattern}
        0 & 2 & \Rest & 1 & \Rest & 0 & 4 \\
        \bar{5} & \Rest & \Rest & 0 & 0 & 0 & 0
    \end{MultiPattern},
    \Color_3 \Color_2 \Color_1 \Color_1 \Color_3},
\end{equation}
\begin{equation}
    \CPat_2 := \Par{\Color_1,
    \begin{MultiPattern}
        1 & \Rest & 0 \\
        0 & \Rest & 1
    \end{MultiPattern},
    \Color_1 \Color_1},
    \quad
    \CPat_3 := \Par{\Color_2,
    \begin{MultiPattern}
        \bar{1} \\
        \bar{1}
    \end{MultiPattern},
    \Color_1},
\end{equation}
\begin{equation}
    \CPat_4 := \Par{\Color_2,
    \begin{MultiPattern}
        0 & 0 \\
        0 & 0
    \end{MultiPattern},
    \Color_1 \Color_1},
    \quad
    \CPat_5 := \Par{\Color_3,
    \begin{MultiPattern}
        0 \\
        0
    \end{MultiPattern},
    \Color_3}.
\end{equation}
Moreover, to interpret the generated multi-patterns, we choose to consider a tempo of $128$
and the rooted scale $\Par{\Scale, \Note{9}{3}}$ where $\Scale$ is the Hirajoshi scale.

Let $\BudSystem := (\Operad, \SetColors, \SetRules, \Color)$ be a bud generating system.
Let $\PartialDerivation$ be the binary relation on $\BudOperad_\SetColors(\Operad)$ such
that
\begin{equation}
    (a, x, u) \PartialDerivation (a, y, v)
\end{equation}
if there is a rule $r\in \SetRules$ and $i \in [|u|]$ such that
\begin{equation}
    (a, y, v) = (a, x, u) \circ_i r.
\end{equation}
An element $x$ of $\Operad$ is \Def{partially generated} by $\BudSystem$ if there is an
element $(\Color, x, u)$ such that $\Par{\Color, \Unit, \Color}$ is in relation with
$(\Color, x, u)$ w.r.t. the reflexive and transitive closure of~$\PartialDerivation$.

For instance, by considering the bud generating
system~\eqref{equ:example_bud_generating_system}, since
\begin{multline}
    \Par{\Color_1,
    \begin{MultiPattern}
        0 \\
        0
    \end{MultiPattern},
    \Color_1}
    \PartialDerivation
    \Par{\Color_1,
    \begin{MultiPattern}
        1 & \Rest & 0 \\
        0 & \Rest & 1
    \end{MultiPattern},
    \Color_1 \Color_1}
    \\
    \PartialDerivation
    \Par{\Color_1,
    \begin{MultiPattern}
        1 & \Rest & 0 & 2 & \Rest & 1 & \Rest & 0 & 4 \\
        0 & \Rest & \bar{4} & \Rest & \Rest & 1 & 1 & 1 & 1
    \end{MultiPattern},
    \Color_1 \Color_3 \Color_2 \Color_1 \Color_1 \Color_3}
    \\
    \PartialDerivation
    \Par{\Color_1,
    \begin{MultiPattern}
        1 & \Rest & 0 & 2 & 2 & \Rest & 1 & \Rest & 0 & 4 \\
        0 & \Rest & \bar{4} & \Rest & \Rest & 1 & 1 & 1 & 1 & 1
    \end{MultiPattern},
    \Color_1 \Color_3 \Color_1 \Color_1 \Color_1 \Color_1 \Color_3},
\end{multline}
the $2$-multi-pattern
\begin{equation}
    \begin{MultiPattern}
        1 & \Rest & 0 & 2 & 2 & \Rest & 1 & \Rest & 0 & 4 \\
        0 & \Rest & \bar{4} & \Rest & \Rest & 1 & 1 & 1 & 1 & 1
    \end{MultiPattern}
\end{equation}
is partially generated by~$\BudSystem$.

The \Def{partial random generation algorithm} is the algorithm defined as follows:
\begin{itemize}
    \item Inputs:
    \begin{enumerate}
        \item A bud generating system
        $\BudSystem := (\Operad, \SetColors, \SetRules, \Color)$;
        \item An integer $k \geq 0$.
    \end{enumerate}
    \item Output: an element of $\Operad$.
\end{itemize}
\begin{enumerate}
    \item Set $x$ as the element $(\Color, \Unit, \Color)$;
    \item Repeat $k$ times:
    \begin{enumerate}
        \item Pick a position $i \in [|x|]$ at random;
        \item If $\SetRules_{\In_i(x)} \ne \emptyset$:
        \begin{enumerate}
            \item Pick a rule $r \in \SetRules_{\In_i(x)}$ at random;
            \item Set $x := x \circ_i r$;
        \end{enumerate}
    \end{enumerate}
    \item Returns $\Prune(x)$.
\end{enumerate}
This algorithm returns an element partially generated by $\BudSystem$ obtained by applying
at most $k$ rules to the initial element $(\Color, \Unit, \Color)$. The execution of the
algorithm builds a composition tree of elements of $\SetRules$ with at most $k$ internal
nodes.

For instance, by considering the bud generating
system~\eqref{equ:example_bud_generating_system}, this algorithm called with $k := 5$ builds
the tree of colored $2$-multipatterns
\begin{equation}
    \begin{tikzpicture}[Centering,xscale=0.22,yscale=0.14]
        \node(0)at(-1.00,-3.75){};
        \node(10)at(7.00,-11.25){};
        \node(11)at(8.00,-11.25){};
        \node(12)at(8.00,-3.75){};
        \node(14)at(11.00,-7.50){};
        \node(2)at(2.00,-7.50){};
        \node(4)at(4.00,-7.50){};
        \node(6)at(4.00,-11.25){};
        \node(7)at(5.00,-11.25){};
        \node(9)at(6.00,-11.25){};
        \node[NodeST](1)at(2.00,-3.75){$\CPat_3$};
        \node[NodeST](13)at(11.00,-3.75){$\CPat_5$};
        \node[NodeST](3)at(5.00,0.00){$\CPat_1$};
        \node[NodeST](5)at(5.00,-3.75){$\CPat_2$};
        \node[NodeST](8)at(6.00,-7.50){$\CPat_1$};
        \draw[Edge](0)--(3);
        \draw[Edge](1)--(3);
        \draw[Edge](10)--(8);
        \draw[Edge](11)--(8);
        \draw[Edge](12)--(3);
        \draw[Edge](13)--(3);
        \draw[Edge](14)--(13);
        \draw[Edge](2)--(1);
        \draw[Edge](4)--(5);
        \draw[Edge](5)--(3);
        \draw[Edge](6)--(8);
        \draw[Edge](7)--(8);
        \draw[Edge](8)--(5);
        \draw[Edge](9)--(8);
        \node(r)at(5.00,3){};
        \draw[Edge](r)--(3);
    \end{tikzpicture}
\end{equation}
which produces the $2$-multi-pattern
\begin{equation}
    \begin{MultiPattern}
        0 & 1 & \Rest & 2 & \Rest & 1 & 3 & \Rest & 2 & \Rest & 1 & 5 & \Rest & 0 & 4 \\
        \bar{5} & \Rest & \Rest & \bar{1} & 0 & \Rest & \bar{4} & \Rest & \Rest & 1 & 1 & 1
            & 1 & 0 & 0
    \end{MultiPattern}.
\end{equation}
Together with the aforementioned interpretation, the generated musical phrase is
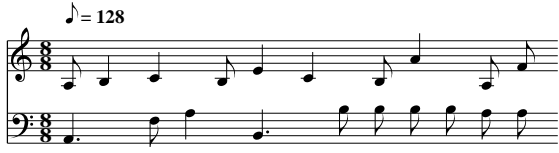
\begin{abc}[name=PhraseExample2,width=.9\abcwidth]
X:1
T:
K:Am
M:8/8
L:1/8
Q:1/8=128
V:voice1
A,1 B,2 C2 B,1 E2 C2 B,1 A2 A,1 F1
V:voice2
A,,3 F,1 A,2 B,,3 B,1 B,1 B,1 B,1 A,1 A,1
\end{abc}

Let $\FullDerivation$ be the binary relation on $\BudOperad_\SetColors(\Operad)$ such that
\begin{equation}
    (a, x, u) \FullDerivation (a, y, v)
\end{equation}
if there are rules $r_1, \dots, r_{|x|} \in \SetRules$ such that
\begin{equation}
    (a, y, v)
    =
    (a, x, u) \circ \Han{r_1, \dots, r_{|x|}}.
\end{equation}
An element $x$ of $\Operad$ is \Def{fully generated} by $\BudSystem$ if there is an element
$(\Color, x, u)$ such that $\Par{\Color, \Unit, \Color}$ is in relation with $(\Color, x,
u)$ w.r.t. the reflexive and transitive closure of~$\FullDerivation$.

For instance, by considering the bud generating
system~\eqref{equ:example_bud_generating_system}, since
\begin{multline}
    \Par{\Color_1,
    \begin{MultiPattern}
        0 \\
        0
    \end{MultiPattern},
    \Color_1}
    \\
    \FullDerivation
    \Par{\Color_1,
    \begin{MultiPattern}
        0 & 2 & \Rest & 1 & \Rest & 0 & 4 \\
        \bar{5} & \Rest & \Rest & 0 & 0 & 0 & 0
    \end{MultiPattern},
    \Color_3 \Color_2 \Color_1 \Color_1 \Color_3}
    \\
    \FullDerivation
    \Par{\Color_1,
    \begin{MultiPattern}
        0 & 1 & \Rest & 2 & \Rest & 1 & \Rest & 0 & 2 & \Rest & 1 & \Rest & 0 & 4 & 4 \\
        \bar{5} & \Rest & \Rest & \bar{1} & 0 & \Rest & 1 & \bar{5} & \Rest & \Rest & 0 & 0
            & 0 & 0 & 0
    \end{MultiPattern},
     \right.
    \\
    \left.
    \Color_3 \Color_1 \Color_1 \Color_1 \Color_3 \Color_2 \Color_1 \Color_1 \Color_3
    \Color_3},
\end{multline}
the $2$-multi-pattern
\begin{equation}
    \begin{MultiPattern}
        0 & 1 & \Rest & 2 & \Rest & 1 & \Rest & 0 & 2 & \Rest & 1 & \Rest & 0 & 4 & 4 \\
        \bar{5} & \Rest & \Rest & \bar{1} & 0 & \Rest & 1 & \bar{5} & \Rest & \Rest & 0 & 0
            & 0 & 0 & 0
    \end{MultiPattern}
\end{equation}
is fully generated by $\BudSystem$.

The \Def{full random generation algorithm} is the algorithm defined as follows:
\begin{itemize}
    \item Inputs:
    \begin{enumerate}
        \item A bud generating system
        $\BudSystem := (\Operad, \SetColors, \SetRules, \Color)$;
        \item An integer $k \geq 0$.
    \end{enumerate}
    \item Output: an element of $\Operad$.
\end{itemize}
\begin{enumerate}
    \item Set $x$ as the element $(\Color, \Unit, \Color)$;
    \item Repeat $k$ times:
    \begin{enumerate}
        \item If all $\SetRules_{\In_i(x)}$, $i \in [|x|]$, are nonempty:
        \begin{enumerate}
            \item Let $\Par{r_1, \dots, r_{|x|}}$ be a tuple of rules such that each $r_i$
            is picked at random in
            $\SetRules_{\In_i(x)}$;
            \item Set $x := x \circ \Han{r_1, \dots, r_{|x|}}$;
        \end{enumerate} 
    \end{enumerate}
    \item Returns $\Prune(x)$.
\end{enumerate}
This algorithm returns an element synchronously generated by $\BudSystem$ obtained by
applying at most $k$ rules to the initial element $(\Color, \Unit, \Color)$. The execution
of the algorithm builds a composition tree of elements of $\SetRules$ of height at most~$k +
1$ wherein the leaves are all at the same distance from the root.

For instance, by considering the bud generating
system~\eqref{equ:example_bud_generating_system}, this algorithm called with $k := 2$ builds
the tree of colored $2$-multipatterns
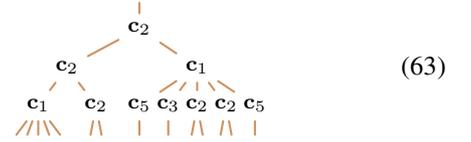
\begin{equation}
    \begin{tikzpicture}[Centering,xscale=0.19,yscale=0.085]
        \node(0)at(0.00,-18.00){};
        \node(1)at(1.00,-18.00){};
        \node(12)at(9.00,-18.00){};
        \node(14)at(11.00,-18.00){};
        \node(16)at(12.50,-18.00){};
        \node(18)at(13.50,-18.00){};
        \node(19)at(14.50,-18.00){};
        \node(21)at(15.50,-18.00){};
        \node(23)at(17.00,-18.00){};
        \node(3)at(2.00,-18.00){};
        \node(4)at(3.00,-18.00){};
        \node(5)at(4.00,-18.00){};
        \node(7)at(5.50,-18.00){};
        \node(9)at(6.50,-18.00){};
        \node[NodeST](10)at(9.00,0.00){$\CPat_2$};
        \node[NodeST](11)at(9.00,-12.00){$\CPat_5$};
        \node[NodeST](13)at(11.00,-12.00){$\CPat_3$};
        \node[NodeST](15)at(13.00,-6.00){$\CPat_1$};
        \node[NodeST](17)at(13.00,-12.00){$\CPat_2$};
        \node[NodeST](2)at(2.00,-12.00){$\CPat_1$};
        \node[NodeST](20)at(15.00,-12.00){$\CPat_2$};
        \node[NodeST](22)at(17.00,-12.00){$\CPat_5$};
        \node[NodeST](6)at(4.00,-6.00){$\CPat_2$};
        \node[NodeST](8)at(6.00,-12.00){$\CPat_2$};
        \draw[Edge](0)--(2);
        \draw[Edge](1)--(2);
        \draw[Edge](11)--(15);
        \draw[Edge](12)--(11);
        \draw[Edge](13)--(15);
        \draw[Edge](14)--(13);
        \draw[Edge](15)--(10);
        \draw[Edge](16)--(17);
        \draw[Edge](17)--(15);
        \draw[Edge](18)--(17);
        \draw[Edge](19)--(20);
        \draw[Edge](2)--(6);
        \draw[Edge](20)--(15);
        \draw[Edge](21)--(20);
        \draw[Edge](22)--(15);
        \draw[Edge](23)--(22);
        \draw[Edge](3)--(2);
        \draw[Edge](4)--(2);
        \draw[Edge](5)--(2);
        \draw[Edge](6)--(10);
        \draw[Edge](7)--(8);
        \draw[Edge](8)--(6);
        \draw[Edge](9)--(8);
        \node(r)at(9.00,5.50){};
        \draw[Edge](r)--(10);
    \end{tikzpicture}
\end{equation}
which produces the $2$-multi-pattern
\begin{small}
\begin{equation}
    \begin{MultiPattern}
        2 & 4 & \Rest & 3 & \Rest & 2 & 6 & \Rest & 2 & \Rest & 1 & \Rest & 0 & 1 & \Rest
            & 2 & \Rest & 1 & \Rest & 1 & \Rest & 0 & 4 \\
        \bar{5} & \Rest & \Rest & 0 & 0 & 0 & 0 & \Rest & 1 & \Rest & 2 & \Rest & \bar{4}
            & \Rest & \Rest & 0 & 1 & \Rest & 2 & 1 & \Rest & 2 & 1
    \end{MultiPattern}.
\end{equation}
\end{small}
Together with the aforementioned interpretation, the generated musical phrase is
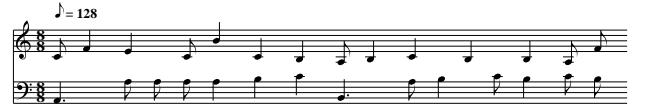
\begin{abc}[name=PhraseExample3,width=1.0\abcwidth]
X:1
T:
K:Am
M:8/8
L:1/8
Q:1/8=128
V:voice1
C1 F2 E2 C1 B2 C2 B,2 A,1 B,2 C2 B,2 B,2 A,1 F1
V:voice2
A,,3 A,1 A,1 A,1 A,2 B,2 C2 B,,3 A,1 B,2 C1 B,2 C1 B,1
\end{abc}

Let $\ColoredDerivation$ be the binary relation on $\BudOperad_\SetColors(\Operad)$ such
that
\begin{equation}
    (a, x, u) \ColoredDerivation (a, y, v)
\end{equation}
if there is a rule $r \in \SetRules$ such that
\begin{equation}
    (a, y, v) = (a, x, u) \ColoredComposition r.
\end{equation}
An element $x$ of $\Operad$ is \Def{colorfully generated} by $\BudSystem$ if there is an
element $(\Color, x, u)$ such that $\Par{\Color, \Unit, \Color}$ is in relation with
$(\Color, x, u)$ w.r.t. the reflexive and transitive closure of~$\ColoredDerivation$.

For instance, by considering the bud generating
system~\eqref{equ:example_bud_generating_system}, since
\begin{multline}
    \Par{\Color_1,
    \begin{MultiPattern}
        0 \\
        0
    \end{MultiPattern},
    \Color_1}
    \ColoredDerivation
    \Par{\Color_1,
    \begin{MultiPattern}
        0 & 2 & \Rest & 1 & \Rest & 0 & 4 \\
        \bar{5} & \Rest & \Rest & 0 & 0 & 0 & 0
    \end{MultiPattern},
    \Color_3 \Color_2 \Color_1 \Color_1 \Color_3}
    \\
    \ColoredDerivation
    \Par{\Color_1,
    \begin{MultiPattern}
        0 & 2 & \Rest & 2 & \Rest & 1 & \Rest & 1 & \Rest & 0 & 4 \\
        \bar{5} & \Rest & \Rest & 0 & 0 & \Rest & 1 & 0 & \Rest & 1 & 0
    \end{MultiPattern},
    \Color_3 \Color_2 \Color_1 \Color_1 \Color_1 \Color_1 \Color_3}
    \\
    \ColoredDerivation
    \Par{\Color_1,
    \begin{MultiPattern}
        0 & 1 & \Rest & 2 & \Rest & 1 & \Rest & 1 & \Rest & 0 & 4 \\
        \bar{5} & \Rest & \Rest & \bar{1} & 0 & \Rest & 1 & 0 & \Rest & 1 & 0
    \end{MultiPattern},
    \Color_3 \Color_1 \Color_1 \Color_1 \Color_1 \Color_1 \Color_3},
\end{multline}
the $2$-multi-pattern
\begin{equation}
    \begin{MultiPattern}
        0 & 1 & \Rest & 2 & \Rest & 1 & \Rest & 1 & \Rest & 0 & 4 \\
        \bar{5} & \Rest & \Rest & \bar{1} & 0 & \Rest & 1 & 0 & \Rest & 1 & 0
    \end{MultiPattern}
\end{equation}
is colorfully generated by $\BudSystem$.

The \Def{colored random generation algorithm} is the algorithm defined as follows:
\begin{itemize}
    \item Inputs:
    \begin{enumerate}
        \item A bud generating system
        $\BudSystem := (\Operad, \SetColors, \SetRules, \Color)$;
        \item An integer $k \geq 0$.
    \end{enumerate}
    \item Output: an element of $\Operad$.
\end{itemize}
\begin{enumerate}
    \item Set $x$ as the element $(\Color, \Unit, \Color)$;
    \item Repeat $k$ times:
    \begin{enumerate}
        \item Pick a rule $r \in \SetRules$ at random;
        \item Set $x := x \ColoredComposition r$;
    \end{enumerate}
    \item Returns $\Prune(x)$.
\end{enumerate}
This algorithm returns an element colorfully generated by $\BudSystem$ obtained by applying
at most $k$ rules to the initial element $(\Color, \Unit, \Color)$. The execution of the
algorithm builds a composition tree of elements of height at most~$k + 1$.

For instance, by considering the bud generating
system~\eqref{equ:example_bud_generating_system}, this algorithm called with $k := 3$ builds
the tree of colored $2$-multipatterns
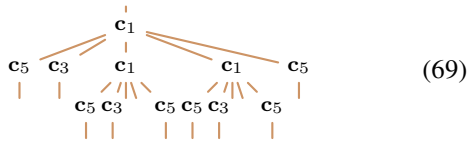
\begin{equation}
    \begin{tikzpicture}[Centering,xscale=0.35,yscale=0.085]
        \node(1)at(0.00,-12.50){};
        \node(10)at(4.00,-12.50){};
        \node(11)at(4.50,-12.50){};
        \node(13)at(5.50,-18.75){};
        \node(15)at(6.50,-18.75){};
        \node(17)at(7.50,-18.75){};
        \node(19)at(8.00,-12.50){};
        \node(20)at(8.50,-12.50){};
        \node(22)at(9.50,-18.75){};
        \node(24)at(10.50,-12.50){};
        \node(3)at(1.50,-12.50){};
        \node(6)at(2.50,-18.75){};
        \node(8)at(3.50,-18.75){};
        \node[NodeST](0)at(0.00,-6.25){$\CPat_5$};
        \node[NodeST](12)at(5.50,-12.50){$\CPat_5$};
        \node[NodeST](14)at(6.50,-12.50){$\CPat_5$};
        \node[NodeST](16)at(7.50,-12.50){$\CPat_3$};
        \node[NodeST](18)at(8.00,-6.25){$\CPat_1$};
        \node[NodeST](2)at(1.50,-6.25){$\CPat_3$};
        \node[NodeST](21)at(9.50,-12.50){$\CPat_5$};
        \node[NodeST](23)at(10.50,-6.25){$\CPat_5$};
        \node[NodeST](4)at(4.00,0.00){$\CPat_1$};
        \node[NodeST](5)at(2.50,-12.50){$\CPat_5$};
        \node[NodeST](7)at(3.50,-12.50){$\CPat_3$};
        \node[NodeST](9)at(4.00,-6.25){$\CPat_1$};
        \draw[Edge](0)--(4);
        \draw[Edge](1)--(0);
        \draw[Edge](10)--(9);
        \draw[Edge](11)--(9);
        \draw[Edge](12)--(9);
        \draw[Edge](13)--(12);
        \draw[Edge](14)--(18);
        \draw[Edge](15)--(14);
        \draw[Edge](16)--(18);
        \draw[Edge](17)--(16);
        \draw[Edge](18)--(4);
        \draw[Edge](19)--(18);
        \draw[Edge](2)--(4);
        \draw[Edge](20)--(18);
        \draw[Edge](21)--(18);
        \draw[Edge](22)--(21);
        \draw[Edge](23)--(4);
        \draw[Edge](24)--(23);
        \draw[Edge](3)--(2);
        \draw[Edge](5)--(9);
        \draw[Edge](6)--(5);
        \draw[Edge](7)--(9);
        \draw[Edge](8)--(7);
        \draw[Edge](9)--(4);
        \node(r)at(4.00,4.69){};
        \draw[Edge](r)--(4);
    \end{tikzpicture}
\end{equation}
which produces the $2$-multi-pattern
\begin{equation}
    \begin{MultiPattern}
        0 & 1 & \Rest & 1 & 2 & \Rest & 2 & \Rest & 1 & 5 & \Rest & 0 & 1 & \Rest & 1
            & \Rest & 0 & 4 & 4  \\
        \bar{5} & \Rest & \Rest & \bar{1} & \bar{5} & \Rest & \Rest & \bar{1} & 0 & 0 & 0
            & \bar{5} & \Rest & \Rest & \bar{1} & 0 & 0 & 0 & 0
    \end{MultiPattern}.
\end{equation}
Together with the aforementioned interpretation, the generated musical phrase is
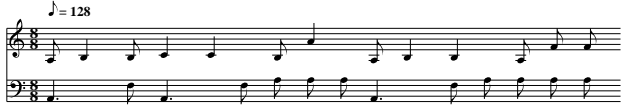
\begin{abc}[name=PhraseExample4,width=1.0\abcwidth]
X:1
T:
K:Am
M:8/8
L:1/8
Q:1/8=128
V:voice1
A,1 B,2 B,1 C2 C2 B,1 A2 A,1 B,2 B,2 A,1 F1 F1
V:voice2
A,,3 F,1 A,,3 F,1 A,1 A,1 A,1 A,,3 F,1 A,1 A,1 A,1 A,1
\end{abc}

\section{Applications: exploring variations of patterns} \label{sec:applications}
We construct here some particular bud generated systems devoted to work with the algorithms
introduced in Section~\ref{subsec:bud_generating_systems}. They generate variations of a
single $1$-multi-pattern $\Pat$ given at input, with possibly some auxiliary data. Each
performs a precise musical transformation of~$\Pat$.

\subsection{Random temporizations}
Given a pattern $\Pat$ and an integer $t \geq 1$, we define the \Def{temporizator bud
generating system} $\BudSystemTemporizator_{\Pat, t}$ of $\Pat$ and $t$ by
\begin{equation}
    \BudSystemTemporizator_{\Pat, t}
    :=
    \Par{\OperadMP_1, \SetColors, \Bra{\CPat_1, \CPat_2, \CPat'_1, \dots, \CPat'_t},
    \Color_1}
\end{equation}
where $\SetColors$ is the set of colors $\Bra{\Color_1, \Color_2, \Color_3}$ and $\CPat_1$,
$\CPat_2$, $\CPat'_1$, \dots, $\CPat'_t$ are the $\SetColors$-colored $1$-multi-patterns
\begin{subequations}
\begin{equation}
    \CPat_1 := \Par{\Color_1, \Pat, \Color_2^{|\Pat|}},
    \quad
    \CPat_2 := \Par{\Color_2, \Pat, \Color_2^{|\Pat|}},
\end{equation}
\begin{equation}
    \CPat'_j :=
    \Par{\Color_2,
    \begin{MultiPattern}
        0 & \Rest^j
    \end{MultiPattern},
    \Color_3},
    \quad j \in [t].
\end{equation}
\end{subequations}
The temporizator bud generating system of $\Pat$ and $t$ generates a version of the pattern
$\Pat$ composed with itself where the durations of some beats have been increased by at most
$t$. The colors, and in particular the color $\Color_3$, prevent multiple compositions of
the colored patterns~$\CPat'_j$, $j \in [t]$, in order to not overly increase the duration
of some beats.

For instance, by considering the pattern
\begin{math}
    \Pat := 0 2 \Rest 1 \Rest 0 4
\end{math}
and the parameter $t := 2$, the partial random generation algorithm ran with the bud
generating system $\BudSystemTemporizator_{\Pat, t}$ and $k := 16$ as inputs produces the
pattern
\begin{equation}
    0 2 \Rest \Rest \Rest 1 \Rest 3 \Rest \Rest \Rest 2 \Rest 1 5 \Rest 0 \Rest \Rest 4
    \Rest.
\end{equation}
Together with the interpretation consisting in a tempo of $128$ and the rooted scale
$\Par{\Scale, \Note{9}{3}}$ where $\Scale$ is the Hirajoshi scale, the generated musical
phrase is
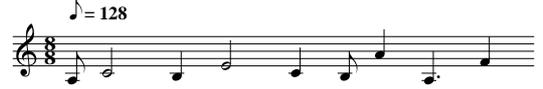
\begin{abc}[name=PhraseExample5,width=.85\abcwidth]
X:1
T:
K:Am
M:8/8
L:1/8
Q:1/8=128
A,1 C4 B,2 E4 C2 B,1 A2 A,3 F2
\end{abc}

\subsection{Random rhythmic variations}
Given a pattern $\Pat$ and a rhythm pattern $\RhyPat$, we define the \Def{rhythmic bud
generating system} $\BudSystemRhythmic_{\Pat, \RhyPat}$ of $\Pat$ and $\RhyPat$ by
\begin{equation}
    \BudSystemRhythmic_{\Pat, \RhyPat}
    :=
    \Par{\OperadMP_m, \SetColors, \Bra{\CPat_1, \CPat_2, \CPat_3}, \Color_1}
\end{equation}
where $\SetColors$ is the set of colors $\Bra{\Color_1, \Color_2, \Color_3}$ and $\CPat_1$,
$\CPat_2$, and $\CPat_3$ are the three $\SetColors$-colored $1$-multi-patterns
\begin{subequations}
\begin{footnotesize}
\begin{equation}
    \CPat_1 := \Par{\Color_1, \Pat, \Color_2^{|\Pat|}},
    \enspace
    \CPat_2 := \Par{\Color_2, \Pat, \Color_2^{|\Pat|}},
    \enspace
    \CPat_3 := \Par{\Color_2, \RhyPat', \Color_3^{|\RhyPat|}},
\end{equation}
\end{footnotesize}
\end{subequations}
where $\RhyPat'$ is the pattern $\Par{0^{|\RhyPat|}, \RhyPat}$. The rhythmic bud generating
system of $\Pat$ and $\RhyPat$ generates a version of the pattern $\Pat$ composed with
itself where some beats are repeated accordingly to the rhythm pattern $\RhyPat$.  The
colors, and in particular the color $\Color_3$, prevent multiple compositions of the colored
pattern~$\CPat_3$. Observe that when $\RhyPat = \epsilon$, each composition involving
$\CPat_3$ deletes a beat in the generated pattern.

For instance, by considering the pattern
\begin{math}
    \Pat := 1 \Rest 0 1 1 \Rest 2
\end{math}
and the rhythm pattern
\begin{math}
    \RhyPat := \Beat \Beat \Rest \Beat \Rest,
\end{math}
the partial random generation algorithm ran with the bud generating system
$\BudSystemRhythmic_{\Pat, \RhyPat}$ and $k := 8$ as inputs produces the pattern

\begin{footnotesize}
\begin{multline}
    2 2 \Rest \Rest 2 \Rest \Rest 1 2 2 \Rest 3 \Rest 1 \Rest 0 1 1 \Rest 2 2 \Rest 1 2 2
    \Rest \Rest 2 \Rest 2 \Rest
    \\
    3 1 \Rest 2 2 \Rest \Rest 2 \Rest.
\end{multline}
\end{footnotesize}
Together with the interpretation consisting in a tempo of $128$ and the rooted scale
$\Par{\Scale, \Note{9}{3}}$ where $\Scale$ is the minor natural scale, the generated musical
phrase is
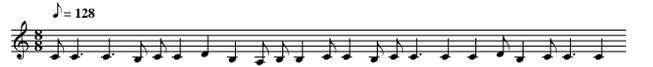
\begin{abc}[name=PhraseExample6,width=1.0\abcwidth]
X:1
T:
K:Am
M:8/8
L:1/8
Q:1/8=128
V:voice1
C1 C3 C3 B,1 C1 C2 D2 B,2 A,1 B,1 B,2 C1 C2 B,1 C1 C3 C2 C2 D1 B,2 C1 C3 C2
\end{abc}

\subsection{Random harmonizations}
For any pattern $\Pat$ and an integer $m \geq 1$, we denote by $[\Pat]_m$ the
$m$-multi-pattern $\Par{[\Pat]_m^{(1)}, \dots, [\Pat]_m^{(m)}}$ satisfying $[\Pat]_m^{(i)} =
\Pat$ for all $i \in [m]$.

Given a pattern $\Pat$ and a degree pattern $\DegPat$ of arity $m \geq 1$, we define the
\Def{harmonizator bud generating system} $\BudSystemHarmonizator_{\Pat, \DegPat}$ of $\Pat$
and $\DegPat$ by
\begin{equation}
    \BudSystemArpeggiator_{\Pat, \DegPat}
    :=
    \Par{\OperadMP_m, \SetColors, \Bra{\CPat_1, \CPat_2, \CPat_3}, \Color_1}
\end{equation}
where $\SetColors$ is the set of colors $\Bra{\Color_1, \Color_2, \Color_3}$ and $\CPat_1$,
$\CPat_2$, and $\CPat_3$ are the three $\SetColors$-colored $m$-multi-patterns
\begin{subequations}
\begin{equation}
    \CPat_1 := \Par{\Color_1, [\Pat]_m, \Color_2^m},
    \quad
    \CPat_2 := \Par{\Color_2, [\Pat]_m, \Color_2^m},
\end{equation}
\begin{equation}
    \CPat_3 :=
    \Par{\Color_2,
    \begin{MultiPattern}
        \DegPat_1 \\
        \DegPat_2 \\
        \vdots \\
        \DegPat_m
    \end{MultiPattern},
    \Color_3}.
\end{equation}
\end{subequations}
The harmonizator bud generating system of $\Pat$ and $\DegPat$ generates an harmonized
version of the pattern $\Pat$ composed with itself, with chords controlled by~$\DegPat$. The
colors, and in particular the color $\Color_3$, prevent multiple compositions of the colored
pattern~$\CPat_3$.

For instance, by considering the pattern
\begin{math}
    \Pat := 2 1 0 2 \Rest 1 \Rest 0 \Rest
\end{math}
and the degree pattern
\begin{math}
    \DegPat := 0 5 \bar{7},
\end{math}
the partial random generation algorithm ran with the bud generating system
$\BudSystemHarmonizator_{\Pat, \DegPat}$ and $k := 3$ as inputs produces the
$3$-multi-pattern
\begin{equation}
    \begin{MultiPattern}
        2 & 1 & 0 & 2 & \Rest & 1 & \Rest & 0 & \Rest \\
        2 & 6 & 5 & 2 & \Rest & 1 & \Rest & 0 & \Rest \\
        2 & \bar{6} & \bar{7} & 2 & \Rest & 1 & \Rest & 0 & \Rest
    \end{MultiPattern}.
\end{equation}
Together with the interpretation consisting in a tempo of $128$ and the rooted scale
$\Par{\Scale, \Note{9}{3}}$ where $\Scale$ is the minor natural scale, the generated musical
phrase is
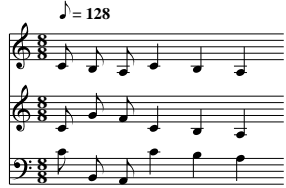
\begin{abc}[name=PhraseExample7,width=.45\abcwidth]
X:1
T:
K:Am
M:8/8
L:1/8
Q:1/8=128
V:voice1
C1 B,1 A,1 C2 B,2 A,2
V:voice2
C1 G1 F1 C2 B,2 A,2
V:voice3
C1 B,,1 A,,1 C2 B,2 A,2
\end{abc}

\subsection{Random arpeggiations}
Given a pattern $\Pat$ and a degree pattern $\DegPat$ of arity $m \geq 1$, we define the
\Def{arpeggiator bud generating system} $\BudSystemArpeggiator_{\Pat, \DegPat}$ of $\Pat$
and $\DegPat$ by
\begin{equation}
    \BudSystemArpeggiator_{\Pat, \DegPat}
    :=
    \Par{\OperadMP_m, \SetColors, \Bra{\CPat_1, \CPat_2, \CPat_3}, \Color_1}
\end{equation}
where $\SetColors$ is the set of colors $\Bra{\Color_1, \Color_2, \Color_3}$ and $\CPat_1$,
$\CPat_2$, and $\CPat_3$ are the three $\SetColors$-colored $m$-multipatterns
\begin{subequations}
\begin{equation}
    \CPat_1 := \Par{\Color_1, [\Pat]_m, \Color_2^m},
    \quad
    \CPat_2 := \Par{\Color_2, [\Pat]_m, \Color_2^m},
\end{equation}
\begin{equation}
    \CPat_3 :=
    \Par{\Color_2,
    \begin{MultiPattern}
        \DegPat_1 & \Rest & \Rest & \dots & \Rest \\
        \Rest & \DegPat_2 & \Rest & \dots & \Rest \\
        \vdots & \vdots & \vdots & \vdots & \vdots  \\
        \Rest & \Rest & \dots & \Rest & \DegPat_m
    \end{MultiPattern},
    \Color_3}.
\end{equation}
\end{subequations}
The arpeggiator bud generating system of $\Pat$ and $\DegPat$ generates an arpeggiated
version of the pattern $\Pat$ composed with itself, where the arpeggio is controlled
by~$\DegPat$. The colors, and in particular the color $\Color_3$, prevent multiple
compositions of the colored pattern~$\CPat_3$. Observe in particular that when $\DegPat =
0^m$, each composition involving $\CPat_3$ creates a repetition of a same beat over the $m$
stacked voices.

For instance, by considering the pattern
\begin{math}
    \Pat := 0 \Rest 2 1 3 \Rest 1
\end{math}
and the degree pattern
\begin{math}
    \DegPat := 0 2 4,
\end{math}
the partial random generation algorithm ran with the bud generating system
$\BudSystemArpeggiator_{\Pat, \DegPat}$ and $k := 8$ as inputs produces the
$3$-multi-pattern
\begin{scriptsize}
\begin{equation}
    \begin{MultiPattern}
        0 & \Rest & 2 & 1 & \Rest & \Rest & 3 & \Rest & 1 & \Rest & 2 & \Rest & 4 & 3 & 5
            & \Rest & 3 & \Rest & \Rest & 1 & \Rest & \Rest & 3 & \Rest & 1 \\
        0 & \Rest & 2 & \Rest & 3 & \Rest & 3 & \Rest & 1 & \Rest & 2 & \Rest & 4 & 3 & 5
            & \Rest & \Rest & 5 & \Rest & \Rest & 3 & \Rest & 3 & \Rest & 1 \\
        0 & \Rest & 2 & \Rest & \Rest & 5 & 3 & \Rest & 1 & \Rest & 2 & \Rest & 4 & 3 & 5
            & \Rest & \Rest & \Rest & 7 & \Rest & \Rest & 5 & 3 & \Rest & 1
    \end{MultiPattern}.
\end{equation}
\end{scriptsize}
Together with the interpretation consisting in a tempo of $128$ and the rooted scale
$\Par{\Scale, \Note{0}{3}}$ where $\Scale$ is the major natural scale, the generated musical
phrase is
\begin{abc}[name=PhraseExample8,width=1.0\abcwidth]
X:1
T:
K:Am
M:8/8
L:1/8
Q:1/8=128
V:voice1
C,2 E,1 D,3 F,2 D,2 E,2 G,1 F,1 A,2 F,3 D,3 F,2 D,1
V:voice2
C,2 E,2 F,2 F,2 D,2 E,2 G,1 F,1 A,3 A,3 F,2 F,2 D,1
V:voice3
C,2 E,3 A,1 F,2 D,2 E,2 G,1 F,1 A,4 C3 A,1 F,2 D,1
\end{abc}

\bibliographystyle{plain}
\bibliography{Bibliography}

\end{document}